\documentclass[pdflatex,sn-mathphys-num]{sn-jnl}% Math and Physical Sciences Numbered Reference Style
\usepackage{graphicx}%
\usepackage{multirow}%
\usepackage{amsmath,amssymb,amsfonts}%
\usepackage{amsthm}%
\usepackage{mathrsfs}%
\usepackage[title]{appendix}%
\usepackage{xcolor}%
\usepackage{textcomp}%
\usepackage{manyfoot}%
\usepackage{booktabs}%
\usepackage{algorithm}%
\usepackage{algorithmicx}%
\usepackage{algpseudocode}%
\usepackage{listings}%

\usepackage{cmsrb}
\usepackage[OT2,T1]{fontenc}

\theoremstyle{thmstyleone}%
\newtheorem{theorem}{Theorem}%  meant for continuous numbers
\theoremstyle{thmstyletwo}%
\newtheorem{example}{Example}%
\newtheorem{remark}{Remark}%

\theoremstyle{thmstylethree}%
\newtheorem{definition}{Definition}%

\usepackage{xspace}
\usepackage{graphicx}
\usepackage{subcaption}
\usepackage{tikz}
\usetikzlibrary{quotes}
\usetikzlibrary{positioning, shapes.geometric}
\usepackage{xcolor} 
\usepackage{cleveref}
\usepackage{todonotes}

\usepackage{amsthm}

\newcommand{\G}{{\mathcal{G}}}
\renewcommand{\L}{{\mathcal{L}}}

\newcommand{\MP}{{\text{MP}}}

\newcommand{\RT}{{RT}}

\newcommand{\R}{\mathbb{R}}
\newcommand{\Q}{\mathbb{Q}}
\newcommand{\Nat}{\mathbb{N}}
\newcommand{\Z}{\mathbb{Z}}
\newcommand{\bbN}{\mathbb{N}}
\newcommand{\bbQ}{\mathbb{Q}}
\newcommand{\bbR}{\mathbb{R}}

\newcommand{\set}[1]{\ensuremath{\left\lbrace #1 \right\rbrace}}
\newcommand{\zug}[1]{\langle #1 \rangle}
\newcommand{\tup}[1]{\langle #1 \rangle}
\newcommand{\floor}[1]{\lfloor #1 \rfloor}
\newcommand{\stam}[1]{}

\renewcommand{\inf}{\mathtt{inf}}

\newcommand{\play}{\mathtt{play}}
\renewcommand{\path}{\mathtt{path}}
\newcommand{\energy}{\mathtt{energy}}
\newcommand{\Min}{\text{Min}\xspace}
\newcommand{\Max}{\text{Max}\xspace}
\newcommand{\prefN}{\pi^{\leq n}}
\newcommand{\pref}[1]{\pi^{\leq #1}}
\newcommand{\Pot}{\text{Pot}}
\newcommand{\St}{\text{St}}
\newcommand{\block}{\sigma^\text{blk}}
\newcommand{\pblock}{\pi^\text{blk}}
\newcommand{\budget}{\sigma^\text{bgt}}
\newcommand{\neig}{\mathcal{N}}
\newcommand{\Bal}{\text{Bal}}
\newcommand{\Mil}{\text{Mil}}

\newcommand{\val}{\mathtt{val}}
\newcommand{\len}{\mathtt{len}}
\newcommand{\xb}{\mathfrak{B}}
\newcommand{\disc}[1]{\mathtt{bif}^{#1}} %there's no such thing as ``mathbf'', also mathtt is nice.

\newcommand{\PO}{Player~$1$\xspace}
\newcommand{\PT}{Player~$2$\xspace}
\newcommand{\PLi}{Player~$i$\xspace}

\newcommand{\blue}{{\color{blue} Blue}\xspace}
\newcommand{\orange}{{\color{orange} Orange}\xspace}
\newcommand{\red}{{\color{red} Red}\xspace}

\newtheorem{lemma}[theorem]{Lemma}
\newtheorem{claim}[theorem]{Claim}

\newtheorem{corollary}[theorem]{Corollary}

\begin{document}

\title{Analyzing the Interaction of Optimal Strategies in Mean-Payoff Bidding Games}


\author[1]{\fnm{Shaull} \sur{Almagor}}\email{shaull@technion.ac.il}

\author[2]{\fnm{Guy} \sur{Avni}}\email{gavni@cs.haifa.ac.il}

\author[1]{\fnm{Julian} \sur{Ewaied}}\email{jolian.ewaied@gmail.com}

\affil[1]{\orgdiv{Department of Computer Science}, \orgname{Technion}}%, \orgaddress{\street{Street}, \city{City}, \postcode{100190}, \state{State}, \country{Country}}}

\affil[2]{\orgdiv{Department of Computer Science}, \orgname{University of Haifa}}%, \orgaddress{\street{Street}, \city{City}, \postcode{10587}, \state{State}, \country{Country}}}

%\affil[3]{\orgdiv{Department}, \orgname{Organization}, \orgaddress{\street{Street}, \city{City}, \postcode{610101}, \state{State}, \country{Country}}}


\abstract{A common assumption when designing an agent in a multi-agent system is that the other agents behave adversarially. 
This allows a designer to obtain the strongest guarantees when they have no control over nor knowledge about the other agents' behavior.
However, when all agents are designed under this adversarial assumption, their actual interaction is not adversarial (e.g., when all players play defensively, no player actually attacks).
In such settings, we would like to know what behavior arises in the multi-agent system.
However, analyzing the interaction among agents is notoriously challenging, both mathematically and algorithmically.
In this paper, we provide such an analysis, focusing on {\em bidding games}, played by two agents on a graph as follows. A token is placed on a vertex, and in each turn an auction (bidding) determines which agent moves the token, thus generating an infinite path that determines the agents' utilities.
We consider {\em mean-payoff} objectives; each vertex is associated with a reward for each player, and the utility in an infinite play is the limit average of the rewards. 
We analyze the play that is generated when each agent follows a strategy that optimizes against an adversary, and consider the two known explicit constructions of optimal strategies. 
The technical challenge stems from the infinitely-many configurations of a bidding game and their complicated dynamics. 
We show that, under some restrictions, the generated play is ultimately periodic, and develop algorithms to compute the players' utilities in it.
}




\maketitle
We consider settings in which a user designs an agent to operate on its behalf. 
For example, advertisers competing for online advertisement slots rely on agents to bid on their behalf~\cite{KN22} (this is in fact a necessity due to the high-frequency of trading in ad-allocation auctions). 
In such settings, users often seek agents with worst-case guarantees. 
This is obtained by modeling the competitors of an agent as adversarial, and the agent follows an optimal strategy in a zero-sum game. 
The setting in which an agent has no knowledge nor does it make assumptions on the competitors that it will play against is often called an ``uncoupled'' setting~\cite{HM03}. 
%We stress that the agents are designed with the purpose of maximizing individual utility with no assumptions on the behavior of their competitors, thus the work on developing agents whose interaction leads to favorable outcomes is not relevant, e.g., increased cooperation among agents~\cite{Ten04,KOC23} or dynamics that converges to equilibria~\cite{BM07}, does not apply to this setting. 

%\footnote{The setting in which there is no knowledge of the other players' utility is often called an ``uncoupled'' setting~\cite{HM03}.} 

We consider a game in which each user knows only their utility, thus all users design their agents to operate against adversarial competitors. 
However, the premise that the agents are in fact adversarial typically does not hold; for example, when all players play defensively, no player actually attacks. 
That is, while each agent provides a worst-case guarantee, these guarantees might be overly pessimistic when the agents turn out not to be purely adversarial. 
{\em Our goal in this paper is to compute the utility that the agents actually obtain in an interaction between agents that are designed to optimize against an adversary. }

%In this paper, we analyze the outcome of an interaction between agents, where each agent is designed under the assumption that it operates against an adversary.  

%\shtodo{Maybe it would be good to repeat (from the abstract) the observation that if all agents are designed this way, then the premise that they are adversarial typically does not hold}

We describe the model on which we investigate this question. {\em Games on graphs} constitute a fundamental model with applications in {\em reactive synthesis}~\cite{PR89} and reasoning about multi-agent systems~\cite{AHK02}, and a deep connection to foundations of logic~\cite{Rab69}. 
%An important application of two-player games on graphs is {\em reactive synthesis}~\cite{PR89}, where the goal is to design a correct by design controller that interacts with an adversarial environment. A winning strategy in the game is a controller implementation that guarantees the specification in any environment. 
%,  and they have deep connections to foundations of logic~\cite{Rab69}. 
A game on a graph proceeds by placing a token on a vertex and the players move it to {\em generate} an infinite path~({\em play}), which determines the utilities of the players. Traditional graph games are {\em turn-based}: the players alternate turns in moving the token. 
We consider two-player {\em bidding games}~\cite{LLPSU99,LLPU96}, which are a class of games on graphs 
%\shtodo{Can you replace ``multi-agent systems'' with something more specific here? Seems weird to send something to AAMAS and give MAS as an example.}
in which an auction (bidding) determines which player acts in each turn: both players are allocated initial budgets that sum up to $1$, and in each turn, they simultaneously submit bids that do not exceed their available budgets, the highest bidder moves the token, and pays his bid to the other player. We focus on {\em mean-payoff} objectives, which are a fundamental quantitative objective in graph games (e.g.,~\cite{ZP96}); each vertex is associated with a reward, and the goal is to maximize the long-run average rewards that are traversed. 
The following example illustrates the setup.

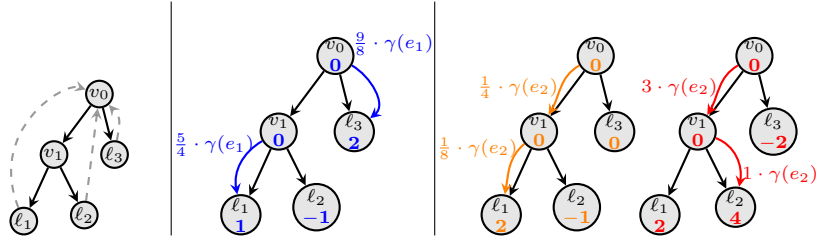
\begin{figure}[ht]
\centering
\begin{subfigure}{0.12\textwidth}
        \footnotesize
         \begin{tikzpicture}[scale=0.8,
          every node/.style={circle, draw, fill=gray!20, minimum size=10pt, inner sep=0pt, font=\bfseries},
          every path/.style={-stealth, thick}
        ]
        
        % Root node
        \node[align=center] (root) at (0.5,-0.5) {$v_0$};
        
        % Child nodes
        \node[align=center] (inner2) at (0.75, -1.5) {$\ell_3$};
        \node[align=center] (inner1) at (-0.25, -1.5) {$v_1$};
        
        % Leaf nodes
        \node[align=center] (leaf1) at (-0.75,-2.6) {$\ell_1$};
        \node[align=center] (leaf2) at (0.25,-2.5) {$\ell_2$};
        
        % Draw the edges
        \draw (root) -- (inner2); %orange
        \draw (root) -- (inner1); %blue
        \draw (inner1) -- (leaf1); %both
        \draw (inner1) -- (leaf2);

%        \draw[blue] (root) edge[out=-30,in=120] (inner2);
        \draw[black, dashed, opacity=0.4] (leaf1) .. controls +(-0.6,1.5) and +(-0.6,0.2) .. (root);
        \draw[black, dashed, opacity=0.4] (leaf2) -- (root);
        \draw[black, dashed, opacity=0.4] (inner2) .. controls +(0.1,0.6) .. (root);
\end{tikzpicture}
 \end{subfigure}%
  \hspace{0.75cm}
 \vline
\begin{subfigure}{0.18\textwidth}
        \footnotesize
         \begin{tikzpicture}[scale=1,
          every node/.style={circle, draw, fill=gray!20, minimum size=10pt, inner sep=0pt, font=\bfseries},
          every path/.style={-stealth, thick}
        ]
        
        % Root node
        \node[align=center] (root) at (0.5,-0.5) {$v_0$\\$\mathbf{ {\color{blue} 0}}$};
        
        % Child nodes
        \node[align=center] (inner2) at (0.75, -1.5) {$\ell_3$\\$\mathbf{{\color{blue} 2}}$};
        \node[align=center] (inner1) at (-0.25, -1.5) {$v_1$\\$\mathbf{{\color{blue} 0}}$};
        
        % Leaf nodes
        \node[align=center] (leaf1) at (-0.75,-2.6) {$\ell_1$\\$\mathbf{{\color{blue} 1}}$};
        \node[align=center] (leaf2) at (0.25,-2.5) {$\ell_2$\\$\mathbf{{\color{blue} -1}}$};
        
        % Draw the edges
        \draw (root) -- (inner2); %orange
        \draw (root) -- (inner1); %blue
        \draw (inner1) -- (leaf1); %both
        \draw (inner1) -- (leaf2);

%        \draw[blue] (root) edge[out=-30,in=120] (inner2);
\draw[blue] (root) edge[out=-30,in=40, "$\frac{9}{8} \cdot \gamma(e_1)$" {pos=0.2, draw=none, fill=none}] (inner2);
        \draw[blue] (inner1) edge[out=-150,in=100, "$\frac{5}{4} \cdot \gamma(e_1)$" {draw=none, left, pos=0.1, fill=none}] (leaf1);
%        \draw[black, dashed, opacity=0.4] (leaf1) .. controls +(-0.6,1.5) and +(-0.6,0.2) .. (root);
 %       \draw[black, dashed, opacity=0.4] (leaf2) -- (root);
  %      \draw[black, dashed, opacity=0.4] (inner2) .. controls +(0.1,1.2) .. (root);
\end{tikzpicture}
 \end{subfigure}%
  \hspace{1cm}
  \vline
 \begin{subfigure}{0.14\textwidth}
         \footnotesize
         \begin{tikzpicture}[scale=1,
          every node/.style={circle, draw, fill=gray!20, minimum size=10pt, inner sep=0pt, font=\bfseries},
          every path/.style={-stealth, thick}
        ]
        
        % Root node
        \node[align=center] (root) at (0.5,-0.5) {$v_0$\\$\mathbf{ {\color{orange} 0}}$};
        
        % Child nodes
        \node[align=center] (inner2) at (0.75, -1.5) {$\ell_3$\\$\mathbf{{\color{orange} 0}}$};
        \node[align=center] (inner1) at (-0.25, -1.5) {$v_1$\\$\mathbf{{\color{orange} 0}}$};
        
        % Leaf nodes
        \node[align=center] (leaf1) at (-0.75,-2.6) {$\ell_1$\\$\mathbf{{\color{orange} 2}}$};
        \node[align=center] (leaf2) at (0.25,-2.5) {$\ell_2$\\$\mathbf{{\color{orange} -1}}$};
        
        % Draw the edges
        \draw (root) -- (inner2); %orange
        \draw (root) -- (inner1); %blue
        \draw (inner1) -- (leaf1); %both
        \draw (inner1) -- (leaf2);

%        \draw[blue] (root) edge[out=-30,in=120] (inner2);
        \draw[orange] (inner1) edge[out=-140,in=80, "$\frac{1}{8} \cdot \gamma(e_2)$" {draw=none, fill=none, pos=0.1, left}] (leaf1);
        \draw[orange] (root) edge[out=-150,in=60, "$\frac{1}{4} \cdot \gamma(e_2)$" {draw=none, fill=none, pos=0.5, left}] (inner1);
%        \draw[black, dashed, opacity=0.4] (leaf1) .. controls +(-0.6,1.5) and +(-0.6,0.2) .. (root);
 %       \draw[black, dashed, opacity=0.4] (leaf2) -- (root);
  %      \draw[black, dashed, opacity=0.4] (inner2) .. controls +(0.1,1.2) .. (root);
\end{tikzpicture}
 \end{subfigure}%
   \hspace{0.75cm}
 \begin{subfigure}{0.2\textwidth}
         \footnotesize
         \begin{tikzpicture}[scale=1,
          every node/.style={circle, draw, fill=gray!20, minimum size=10pt, inner sep=0pt, font=\bfseries},
          every path/.style={-stealth, thick}
        ]
        
        % Root node
        \node[align=center] (root) at (0.5,-0.5) {$v_0$\\$\mathbf{ {\color{red} 0}}$};
        
        % Child nodes
        \node[align=center] (inner2) at (0.75, -1.5) {$\ell_3$\\$\mathbf{{\color{red} -2}}$};
        \node[align=center] (inner1) at (-0.25, -1.5) {$v_1$\\$\mathbf{{\color{red} 0}}$};
        
        % Leaf nodes
        \node[align=center] (leaf1) at (-0.75,-2.6) {$\ell_1$\\$\mathbf{{\color{red} 2}}$};
        \node[align=center] (leaf2) at (0.25,-2.5) {$\ell_2$\\$\mathbf{{\color{red} 4}}$};
        
        % Draw the edges
        \draw (root) -- (inner2); %orange
        \draw (root) -- (inner1); %blue
        \draw (inner1) -- (leaf1); %both
        \draw (inner1) -- (leaf2);

%        \draw[blue] (root) edge[out=-30,in=120] (inner2);
        \draw[red] (inner1) edge[out=-20,in=80, "$1 \cdot \gamma(e_2)$" {draw=none, fill=none, pos=0.8, right}] (leaf2);
        \draw[red] (root) edge[out=-150,in=60, "$3 \cdot \gamma(e_2)$" {draw=none, fill=none, pos=0.5, left}] (inner1);
 %       \draw[black, dashed, opacity=0.4] (leaf1) .. controls +(-0.6,1.5) and +(-0.6,0.2) .. (root);
   %     \draw[black, dashed, opacity=0.4] (leaf2) -- (root);
     %   \draw[black, dashed, opacity=0.4] (inner2) .. controls +(0.1,1.2) .. (root);
\end{tikzpicture}
 \end{subfigure}%
 \caption{Left: an arena. Middle to right: respectively, the rewards and block strategy of \blue, \orange, and \red.} 
\label{fig:example}
\end{figure}    

%The following example illustrates the setup.
\begin{example}
\label{ex:intro}
Consider a game played on the arena (graph) depicted in Fig.~\ref{fig:example} Left.
The Middle figure, depicts the reward that user \blue obtains in each vertex. \blue designs their agent with no assumptions on the competitors that it will encounter; \blue constructs a zero-sum game in which the competitors are adversarial, and applies an off-the-shelf algorithm~\cite{AHC19} to construct an optimal strategy. The strategy is depicted in Fig.~\ref{fig:example} Middle: each vertex is associated with a bid written beside the vertex, and a \blue arrow depicts the choice of successor upon winning the bidding (e.g., if \blue wins the bidding at $v_1$, he proceeds to $\ell_1$). The construction in~\cite{AHC19} defines bids in a specific form: a vertex-dependent constant multiplied by a {\em normalization} (depicted $\gamma(\cdot)$, in Fig.~\ref{fig:example}), which depends on the history of rewards observed as we detail in Sec.~\ref{sec:block}. 

\blue can be paired with two possible competitors, \orange and \red, who differ in the rewards that they obtain in the vertices. Both users design their agents with the same adversarial assumption that \blue applies; they construct a zero-sum game and apply an off-the-shelf algorithm to find an optimal strategy. Both the rewards and the strategies that they construct are depicted in Fig.~\ref{fig:example} Right.

We illustrate how two strategies generate a play. Suppose first that \blue is paired with \orange. For simplicity, assume $a = \gamma(e_1) = \gamma(e_2)$. 
At $v_0$, the bids are $b_\blue = 9/8 \cdot a > 1/4 \cdot a = b_\orange$, thus \blue wins the bidding, pays $b_\blue$ to \orange, and moves the token to $\ell_3$. \blue obtains a reward of $2$ and \orange a reward of $0$, and the game returns to $v_0$. Since the collected rewards differ, this time the normalizations in the bids might differ, leading to a possible different bidding outcome. 
%We point out that the generated play has an involved structure. Indeed, the choice of normalization is such that it decreases as the energy increases, thus eventually, \orange will win the bidding at $v_0$ and choose to proceed to $v_1$. 
Suppose now that \blue is paired with \red. At $v_0$, the bids are $b_\blue = 9/8 \cdot a<  3\cdot a = b_\red$. Now \red wins, and moves the token to $v_1$. At $v_1$, the bids are $b_\blue = 5/4 \cdot a > 1 \cdot a = b_\orange$, and \blue moves the token to $\ell_1$. %We point out that an agent cannot win all biddings since this will exhaust their budget. This is where the normalization comes into the picture; roughly, as the energy increases, the normalization and the bids decrease. 

%Indeed, \red's strategy is a possible strategy of \blue's competitor, thus the composition of \blue and \red generates a play in \blue's zero-sum game, and \blue's guarantee applies to any such play, and similarly for the other agents. 
Fig.~\ref{fig:plots} depicts the accumulated rewards for each of the players in the two generated plays. 
The primary question that we study is the regularity of the generated play. The plots depict a non-trivial structure of the generated play, and hints at the technical difficulty of this problem. A careful inspection of the plots reveals a regular behavior. We prove formally that this is a general phenomenon: the generated play is ultimately periodic in all {\em recurrent} strongly-connected games (see Thm.~\ref{thm:budget-lasso}). 
In terms of utility, while the strategies provide the same modest worst-case guarantees,  a non-negative mean-payoff reward, the actual payoffs are way higher: the utility (mean-payoff reward) when \blue plays against \orange are both $0.5$ and when \blue plays against \red they are $0.411$ and $0.47$. The higher \blue utility in the former could stem from the agreement between \blue and \orange at $v_1$ whereas \blue disagrees with \red everywhere. %We stress that these agreements are not known when the strategies are constructed and the worst-case guarantees apply to any competitor. 
\begin{figure}[ht]
\centering
\includegraphics[width=4.2cm]{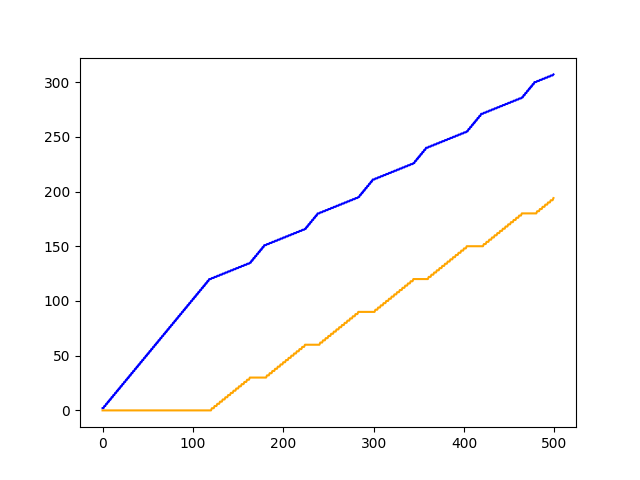}
\includegraphics[width=4.2cm]{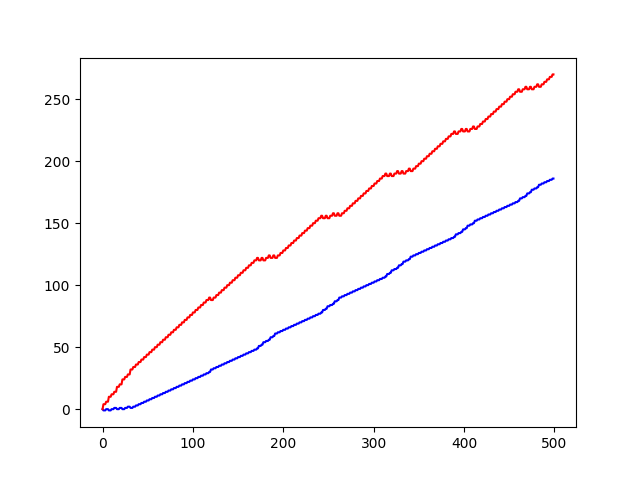}
\caption{Left: \blue vs \orange, Right: \blue vs \red, Y-axis: accumulated rewards (energy), X-axis: time}
\label{fig:plots}
\end{figure}
\end{example}

%In particular, with mean-payoff objectives (on strongly connected graphs) we are able to produce an output strategy even when $\Th_1(v_0) + \Th_2(v_0) \geq 1$ (as opposed to the qualitative case).
%%%%%%%%%%%%%%%%%%%%%%%%%%

\noindent{\bf Motivation.}
We describe concrete applications. 

{\bf Fair allocation of resources} is a timely topic (e.g.,~\cite{ABFV22,ALMW22}). The goal is to allocate a collection of items $\set{i_1,\ldots, i_k}$ to agents in a {\em fair} manner. 
Bidding games have been applied as a natural and useful mechanism for fair allocation~\cite{MKT18,BEF21}: each agent is allocated an initial {\em scrip} budget, i.e., a budget that is used for the purpose of the mechanism and carries no utility for the agent, and the resources are auctioned sequentially. %, i.e., on the $j$-th turn, Item~$i_j$ is auctioned to the agents, for $1 \leq j \leq k$.\shtodo{remove ``$for 1\le j\le k$'', I think. Maybe even everything after ``sequentially''.} 
It is common practice to construct agents that provide worst-case guarantees on the fraction of the utility proceeds by solving a zero-sum game (e.g.,~\cite{GBI21}). 

Our work gives rise to a mechanism for repeated allocation of resources, as the following example illustrates. 
A librarian needs to decide, each week, which newspaper it purchases, either New York Times (NYT) or Wall Street Journal (WSJ). The librarian applies the mechanism depicted in Fig.~\ref{fig:example} Left, in which two students participate: a first bidding at $v_0$ determines whether NYT is purchased ($\ell_3$) and otherwise, a second bidding at $v_1$ determines whether WSJ is purchased ($\ell_1$) or whether no newspaper is purchased ($\ell_2$). Student \blue likes both NYT and WSJ with a preference for the former, Student \orange only likes WSJ, and both students suffer from if no newspaper is purchased. 
Analyzing the behavior of the game allows the librarian to allocate the correct budget for newspapers, rather than use estimates based on the worst-case guarantees of the two players.

Similar mechanisms can be applied for fairly allocating daily computing time to users (e.g., on a GPU) or fairly allocating an advertisement slot among two advertisers as we elaborate below.

{\bf Auction-based scheduling: } Consider the problem of finding a plan (a path) in a graph that models an environment for a conjunction of two objectives. For example, consider the task of designing a plan for a patrolling robot that needs to maximize the time it spends in two locations $t_1$ and $t_2$. 
A {\em decoupled} approach to planning~\cite{AH+26} independently constructs a policy for each objective and composes the policies at runtime. Advantages of the approach include enabling parallel computing and modularity, e.g., if $t_1$ changes, only $\sigma_1$ needs to be updated and $\sigma_2$ can stay fixed. Runtime composition is challenging due to the need to resolve conflicts. {\em Auction-based scheduling} (ABS)~\cite{AMS24} applies an auction in each turn to determine which policy chooses the next action. 

%This decoupled approach has several benefits: (1)~it is parallelizable; the independent construction means that each strategy can be constructed on a separate CPU, (2)~it is modular; e.g., if $t_1$ changes, only $\sigma_1$ needs to be updated and $\sigma_2$ can stay fixed, and (3)~it enables an iterative design procedure: suppose that initially, only $t_1$ is known for which $\sigma_1$ is found, then afterwards $t_2$ is added for which $\sigma_2$ is found and composed with $\sigma_1$ at runtime. 
%Recently, a \emph{decoupled} approach for planning, called {\em auction-based scheduling} (ABS), was developed by~\cite{AMS24}. 
ABS proceeds as follows. 
For each target $t_i$, construct a bidding game $\G_i$ on the input graph, where in each turn, the players bid for who moves the robot. The game $\G_i$ is a zero-sum game in which the proponent aims to spend as much time at $t_i$. This is naturally expressed as a mean-payoff objective. 
Solve each game to obtain optimal strategies $\sigma_1,\sigma_2$ for the robot. Each $\sigma_i$ is accompanied by a guarantee $c_i$ on the time spent at $t_i$, for $i \in \set{1,2}$. 
Then, a plan for the robot is obtained by composing the two strategies by letting them play against one another. It is not hard to see that the worst-case guarantees apply, namely time that the robot spends at $t_i$ is at least $c_i$. 
Crucially, the guarantees of the strategies may be overly pessimistic. In this work, we seek to compute the actual time the robot spends in the targets.

Finally, we point out that~\cite{AH+26,AMS24} only consider qualitative objectives. Our work is the first to consider quantitative objectives, which is particularly appealing since it allows to quantify and weigh the (possible) drop in performance with the advantages of decoupling. 

%Both can be naturally specifies as a mean-payoff objective. 

{\bf Advantages over Nash equilibrium (NE)}. The standard solution concept for non-zero-sum games is NE. We point to advantages of the solution that we consider. 
First, unlike NE, where each agent knows the objectives of the competitor as well as assumes that they rationally aim to maximize it, here, agents assume neither: they have no knowledge on the objective of their competitors and make no assumptions on their behavior.
Second, a shortcoming of NE is that it is not clear how agents end up playing an NE. One either needs to assume player dynamics, which does not always converge to NE, or assume a centralized authority that can prescribe strategies to agents, which is arguably a strong assumption, not to mention that finding an NE is computationally intractable~\cite{DGP06}. We find it a feature of our solution concept that it neither relies on dynamics nor requires a centralized authority.

\smallskip
\noindent{\bf Our results.}
We consider two-player mean-payoff bidding games played on strongly-connected graphs. The game is not zero sum, but each agent only knows their individual objective. Each agent assumes the competitor is adversarial, and follows an optimal strategy. It is known that $\epsilon$-optimal pure strategies exist in mean-payoff bidding games. Moreover, there are two known explicit constructions of $\epsilon$-optimal strategies: the {\em block strategy} $\block$~\cite{AHC19} and the {\em budget strategy} $\budget$~\cite{AJZ21}. 

The first question that we study regards the regularity of the path that is generated when two block strategies or two budget strategies play against each other. 
This is highly non-trivial since a bidding game contains infinitely-many (in fact, uncountably-many) configurations, and the strategies are described succinctly; $\block$ chooses bids depending on the accumulated weights and $\budget$ chooses bids depending on the current budget. 
One can argue that establishing regularity of the generated path is a prerequisite result; 
it is hard to imagine, e.g., an algorithm to predict the payoffs of the generated play, if the play is not ultimately periodic. 

%In practice, an immediate approach to assess the performance of a generated play is to simply simulate the strategies. However, without our ultimately-periodic result, the designer could not draw any conclusions from such a simulation since the payoffs could very well fluctuate arbitrarily. Our result implies that only a finite simulation is required since every simulation eventually stabilizes.

%:  and on a (possibly irrational) number in $[0,1]$. 

{\bf Recurrent games}. {\em Repeated games} are simple, common, and important; e.g., repeated prisoner's dilemma is extensively studied (e.g.,~\cite{KM+82}) and a repeated mean-payoff game is studied in the seminal book~\cite{AMS95}. In a {\em repeated bidding game}, winning a bidding entails a reward and losing a bidding entails a penalty, and the goal is to maximize the mean-payoff reward. For example, suppose that two advertisers participate in a daily bidding to determine which of the two ads displays that day, each advertiser is rewarded $1$ whenever they win a bidding and $0$ otherwise, then an advertiser's mean-payoff reward represents the long-run ratio of the biddings won, which can also be thought of as their ``visibility'', e.g., the number of days in a year that an ad is displayed. 
{\em Recurrent games} generalize repeated games. They are played on a tree in which the leaves point to the root (see Fig.~\ref{fig:example} Left). A repeated bidding game is the special case of a tree of height $1$. Recurrent games arise naturally, as illustrated in the mechanism for repeated fair resource allocation above. Moreover, recurrent bidding games served as a useful stepping stone towards solutions to general games; e.g., this was the case in~\cite{AJZ21}.

For block strategies, we show that the path $\pblock$ generated by two block strategies is ultimately periodic. Observe that the plots depicted in Fig.~\ref{fig:plots} do indeed hint that the paths are ultimately periodic, and it is interesting to note that the period, particularly in the \blue vs \red plot, is highly non-trivial. 
Our proof gives rise to an algorithm to compute the payoffs of $\pblock$. 
%Another implication of our proof is that, interestingly, under mild assumptions, the utilities in the generated play coincide. For example, in Ex.~\ref{ex:intro}, we have $\MP_1(\pblock) = \MP_2(\pblock) = 0.5$. This constitutes additional motivation for using bidding as a mechanism for ongoing fair allocation of resources. 
For budget strategies, we devise a sound algorithm to reason about the generated path in recurrent games and show that in repeated bidding games,  the generated path is ultimately periodic. 

%The bowtie game is a fundamental game; a recurrent game with a root and two leaves that point to it. Practically, it corresponds to a repeated stateless auction, where each player simply aims to maximize the ratio of the biddings won. 
\smallskip
Our conceptual contribution --- analyzing the interaction between optimal strategies --- entails a careful mathematical analysis. Moreover, it is (crucially) tailored to the structure of the optimal strategies being used. Thus, the paper must first recall some technical aspects of these strategies. We defer most of the technicalities to the appendix, but still maintain analysis since it has merit in its own right; both the results themselves and illustrating the techniques, which are both novel to bidding games and will are challenging to extend beyond the settings that we consider.

\noindent{\bf Related work.}
To the best of our knowledge, analyzing the outcome of two optimal strategies has not been studied before. A conceptually similar line of work, which has recently attracted considerable attention, analyzes the outcome of a repeated game, where players follow a fixed online-learning algorithm. 
For example,~\cite{BP18} shows that even though it is guaranteed that the time-average of the sequence of strategy profiles generated converges to an NE, surprisingly, the sequence itself does not converge to an NE, and in fact tends away from it. 
Similar to our eventually-periodic structural results,~\cite{MPP18} establishes ``recurrence'' in the sequence of strategy profiles. 
In both our and these works, the analysis is technically challenging since it entails analyzing a dynamical system; in our case, a piece-wise linear dynamical system.

%%%%%%%%%%%%%%%%%%%%%%%%%%%%%%%%%%%%%%%%%%%%%%%%%%%%%%%%%%%%%%%%%%%%%%%%

\section{Preliminaries}

A \emph{mean-payoff bidding game} is a tuple $\G = \zug{V, E, w_1,w_2}$, where $V$ is a set of vertices, $E \subseteq V \times V$ is a set of edges, for $i=1,2$, $w_i: V \rightarrow \Q$ is a weight function for \PLi. We restrict attention to strongly-connected games, i.e., $\zug{V, E}$ is a strongly-connected graph. The {\em neighbors} of $v \in V$ are $\neig(v) = \set{u: E(v,u)}$. 
A game is {\em zero-sum} when $w_1 = -w_2$, i.e., the reward of \PO is the penalty for \PT. In such games, we refer to \PO as Max, \PT as Min, and use only one weight function $w=w_1$. 

A \emph{configuration} of a bidding game is a tuple $\zug{v, B}$ meaning that the token is placed on $v \in V$ and \PO's budget is $B\in [0,1]$. We normalize the sum of budgets to $1$, thus implicitly \PT's budget is $1-B$. 
A \emph{strategy} for \PLi is a ``recipe'' that, given the history of the game, specifies an {\em action} for the player to take. Formally, a strategy for \PLi is a function $\sigma_i: (V\times [0,1])^+ \rightarrow V \times [0,1]$, which given a history of configurations, returns an action $\zug{u, b} \in (V \times [0,1])$ meaning that \PLi bids $b$ and moves the token to $u$ upon winning the bidding. 
We restrict to {\em legal} strategies that choose (1)~a bid that does not exceed the available budget and (2)~move the token to a neighboring vertex. 
%short: this only talks about P1:
%We restrict to {\em legal} strategies that at configuration $\zug{v, B}$, assign actions of the form $\zug{u, b}$ where (1)~the bid does not exceed the available budget, i.e., $b \leq B$ and (2)~$u$ is a neighbor of $v$, i.e., $u \in \neig(v)$. 
An initial configuration $c_0 \in (V \times [0,1])$ and two strategies $\sigma_1$ and $\sigma_2$ give rise to a unique {\em play}, denoted $\play(c_0, \sigma_1, \sigma_2)$, which is an infinite sequence of configurations that is defined inductively as follows. 
For $n \in \Nat$, denote the \emph{$n$-th prefix} of an infinite play $\pi = c_0, c_1,\ldots$ by $\prefN = c_0,\ldots,c_n$. 
The first configuration of $\play(c_0, \sigma_1, \sigma_2)$ is $c_0$. Suppose that $\pref{j}$ is defined and ends in configuration $c_j = \zug{v_j, B_j}$. We define the next configuration $c_{j+1}$ as follows. For $i\in \set{1,2}$, let $\zug{u_i, b_i} = \sigma_i(c_0, \ldots, c_j)$ be \PLi's next action. If $b_1 \geq b_2$, \PO wins the bidding 
and the next configuration is $c_{j+1} = \zug{u_1, B_j-b_1}$, and otherwise \PT wins the bidding and $c_{j+1} = \zug{u_2, B_j + b_2}$. Note that we break bidding ties in favor of \PO, and our analysis can easily be adapted to accommodate more involved tie-breaking mechanisms such as alternating turns, but the issue of tie breaking is orthogonal to the questions that we study.
%This  choice is insignificant: the worst-case guarantees of the strategies apply to any tie-breaking mechanism and the effect of the mechanism on our analysis is very minor. 
The path in $\G$ that \emph{corresponds} to $\play(c_0, \sigma_1, \sigma_2) = \zug{v_0, B_0}, \zug{v_1, B_1}, \ldots$ is denoted $\path(c_0, \sigma_1, \sigma_2) = v_0, v_1, \ldots$. 
%\shtodo{Some keyword needs to be defined here (either \emph{corresponds}, or define something like \emph{induced path})}
%\gutodo{I forgot the notation I wanted to introduce. I added it now. Is this what you meant?}
%\shtodo{Great.}
%We omit $\G$ and $c_0$ when they are clear from the context.

\begin{definition}\label{def:MP}{\bf (Mean-payoff and energy).}
Let $\prefN = c_0,\ldots,c_n$ with $c_j = \zug{v_j, B_j}$, for $n \in \Nat$ and $0 \leq j \leq n$. %short We associate a pair of {\em energies} with $\prefN$: 
For $i \in \set{1,2}$, define \PLi's energy in $\prefN$ to be $\energy_i(c_0,\ldots, c_n) = \sum_{0 \leq j < n} w_i(v_j)$. \PLi's {\em payoff} in $\pi$ is $\MP_i(\pi) = \liminf\limits_{n \rightarrow \infty} \frac{1}{n}\energy_i(\prefN)$.
%We associate two payoffs with $\pi$. For $i \in \set{1,2}$, define $\MP_i(\pi) = \liminf_{n \rightarrow \infty} \frac{\energy_i(\prefN)}{n}$. 
When $\G$ is zero-sum, we write $\MP(\pi)$ for Max's reward. 
\end{definition}

\paragraph*{The mean-payoff value in zero-sum games}
The guarantees that the strategies provide arise from results on zero-sum mean-payoff bidding games, which we briefly survey next. 
Consider a zero-sum game $\G$. The {\em mean-payoff value} of $\G$, denoted $\val(\G)$, is intuitively the optimal reward that Max can ensure. Quite surprisingly, it was shown in \cite{AHC19} that the $\val(\G)$ does not depend on the initial budgets and only on the structure of the game. 

%; namely, for every $\epsilon > 0$, roughly, the payoff that Max can ensure with a budget of $\epsilon$ and $1-\epsilon$ is the same. Moreover, the value is characterized as follows. The {\em random-turn game} that corresponds to a game $\G$ and a probability $p$, denoted $\RT(\G, p)$, is a game in which in each turn, instead of bidding, \PO moves with probability $p$, and \PT moves otherwise. Formally, $\RT(\G, p)$ is a stochastic game~\cite{Con92}. Such games are known to have a {\em value} (see~\cite{Put05}), denoted $\MP\big(\RT(\G, p)\big)$, which is intuitively the expected long-run average of the weights under optimal play of the players. When $p = \frac{1}{2}$, we omit it and write $\RT(\G)$.

\begin{theorem}\label{thm:zero-sum-MP}
\cite{AHC19}
Consider a strongly-connected zero-sum mean-payoff bidding game $\G=\zug{V, E, w}$, an initial vertex $v_0$, and $\epsilon,\epsilon' > 0$. There exists a value $\val(\G)$ such that 
\begin{itemize}
\item Max can guarantee payoff $\val(\G)-\epsilon$ with any positive budget. Formally, 
Max has a strategy $\sigma$ s.t. for any Min strategy $\tau$, ensures $\MP(\play(\zug{v_0, \epsilon'}, \sigma, \tau)) \geq \val(\G) - \epsilon$. 
\item Max cannot do better even with (almost) all the budget. Formally, there is a Min strategy $\tau$ such that for every Max strategy $\sigma$, we have $\MP(\play(\zug{v_0, 1-\epsilon'}, \sigma, \tau)) \leq \val(\G) + \epsilon$. 
\end{itemize}
\end{theorem}

A corollary of Thm.~\ref{thm:zero-sum-MP} is that two optimal strategies are composable in the following sense. Consider a game $\G = \zug{V, E, w_1, w_2}$, let $\G_i = \zug{V, E, w_i}$, for $i \in \set{1,2}$, be two zero-sum mean-payoff games, and $\sigma_i$ be an $\epsilon$-optimal strategy for Max in $\G_i$. Then, since $\sigma_i$ requires only a positive initial budget to guarantee $\val(\G_i)-\epsilon$, we can compose $\sigma_1$ with $\sigma_2$ by allocating, for example, a budget of $0.5$ to each. Trivially, for $\pi = \play(\zug{v_0, 0.5}, \sigma_1, \sigma_2)$, we have both $\MP_1(\pi) \geq \val(\G_1) - \epsilon$ and $\MP_2(\pi) \geq \val(\G_2)-\epsilon$.

We focus on games played on {\em recurrent graphs} (see Fig.~\ref{fig:example}), defined as follows. 
\begin{definition}
{\bf (Recurrent graphs).}
A {\em recurrent} graph $\tup{V,E}$ is a strongly-connected graph given as a tree with root $v_0$, in which all leaves point to $v_0$, thus all cycles traverse $v_0$.\footnote{This restriction can be lifted to DAGs by duplicating vertices that are shared between paths that lead to the same leaf.}
We call a path from root to leaf a \emph{vine} and a visit to $v_0$ a \emph{milestone}.  
For a play $\pi = c_0, c_1,\ldots$ let $\Mil = \set{n_1,n_2,\ldots}$ be the indices in which $\pi$ visits the root, i.e., for $n \in \Mil$, we have $c_n= \zug{v_0, B_n}$ for some $B_n$. For $i \geq 1$, the path $v_{n_i}\ldots v_{n_{i+1}-1}$ is a vine. 
\end{definition}

In the next two sections we describe two explicit constructions of $\epsilon$-optimal strategies in zero-sum games and analyze the play that arises from composing two strategies of the same type in a non-zero sum game.

%%%%%%%%%%%%%%%%%%%%%%%%%%%%%%%%%%%%%%%%%%%%%%%
\section{Analyzing the Play Generated by Two Block Strategies}
In this section, we describe the construction of {\em block strategies}~\cite{AHC19} and analyze the play generated by two such strategies. We show that in recurrent games, the path that the play forms is ultimately periodic.
 The proof gives rise to an algorithm that computes the payoff of the play. Finally, we identify a class of games for which we show that the play traverses at most two vines infinitely often.

\subsection{The block strategy}
\label{sec:block}
We describe the parts of the construction of a block strategy that are essential for this work. The full construction and its correctness proof is described in \cite[Sec.~4.2.1]{AHC19}.

The strategy chooses $e_0 \in \Nat$, ``pretends'' that the initial energy is $e_0$, and ensures that the energy never drops to $0$. It is not hard to see that this guarantees a non-negative mean-payoff. Note that in order to guarantee payoff $c$, one can simply follow a strategy in a game in which all weights are decreases by $c$. 

The strategy determines bids according to the accumulated rewards, defined as follows. Recall that $\Mil$ are the indices of the play that visit the root. For $n \in \Mil$, the {\em virtual energy} of $\prefN$, denoted $e(\prefN)$, can be thought of as the ``current energy''; it is the change in energy  added to the initial energy, $e(\prefN) = e_0 + \energy(\prefN)$.

\smallskip
\noindent{\bf Strategy structure.}
Suppose that the game is at vertex $v$ following prefix $\pref{m}$, and let $n \leq m$ with $n \in \Min$ be the last visit to the root. 
Then, the strategy chooses $\zug{b, v'}$, where $v'$ is a neighbor of $v$, and the bid $b$ is of the form $\St(v) \cdot \gamma\big(e(\prefN)\big)$, where $\St(v)$ is called the {\em strength} of $v$ and $\gamma(\cdot)$ is a {\em normalization scheme} that depends on the virtual energy. 
\smallskip

Importantly, the choices of $v'$ and $\St(v)$ are pre-computed and {\em do not} depend on the prefix of the play. The analysis in this paper only requires that they are constant. 
For completeness, we illustrate the definition in App.~\ref{app:strengths}.

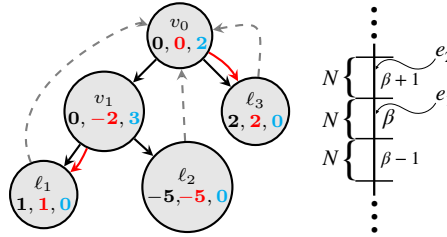
\begin{figure}[ht]
\centering
\begin{subfigure}{0.2\textwidth}
        \footnotesize
         \begin{tikzpicture}[scale=1,
          every node/.style={circle, draw, fill=gray!20, minimum size=10pt, inner sep=0pt, font=\bfseries},
          every path/.style={-stealth, thick}
        ]
        
        % Root node
        \node[align=center] (root) at (0.5,-0.5) {$v_0$\\$\mathbf{{\color{black} 0}, {\color{red} 0}, {\color{cyan} 2}}$};
        
        % Child nodes
        \node[align=center] (inner2) at (1.5, -1.5) {$\ell_3$\\$\mathbf{{\color{black} 2}, {\color{red} 2}, {\color{cyan} 0}}$};
        \node[align=center] (inner1) at (-0.5, -1.5) {$v_1$\\$\mathbf{{\color{black} 0}, {\color{red} -2}, {\color{cyan} 3}}$};
        
        % Leaf nodes
        \node[align=center] (leaf1) at (-1.3,-2.6) {$\ell_1$\\$\mathbf{{\color{black} 1}, {\color{red} 1}, {\color{cyan} 0}}$};
        \node[align=center] (leaf2) at (0.6,-2.5) {$\ell_2$\\$\mathbf{{\color{black} -5},\!\! {\color{red} -5}, {\color{cyan} 0}}$};
        
        % Draw the edges
        \draw (root) -- (inner2); %orange
        \draw (root) -- (inner1); %blue
        \draw (inner1) -- (leaf1); %both
        \draw (inner1) -- (leaf2);

        \draw[red] (root) edge[out=-30,in=120] (inner2);
%        \draw[orange] (root) edge[out=-150,in=60] (inner1);
        \draw[red] (inner1) edge[out=-115,in=40] (leaf1);
%        \draw[orange] (inner1) edge[out=-140,in=65] (leaf1);
        \draw[black, dashed, opacity=0.5] (leaf1) .. controls +(-0.6,1.5) and +(-0.6,0.2) .. (root);
        \draw[black, dashed, opacity=0.5] (leaf2) -- (root);
        \draw[black, dashed, opacity=0.5] (inner2) .. controls +(0.1,1.2) .. (root);
        
        \end{tikzpicture}
       \end{subfigure}
       \hspace{1cm}
\begin{subfigure}{0.2\textwidth}
\centering
        \includegraphics[height=3cm]{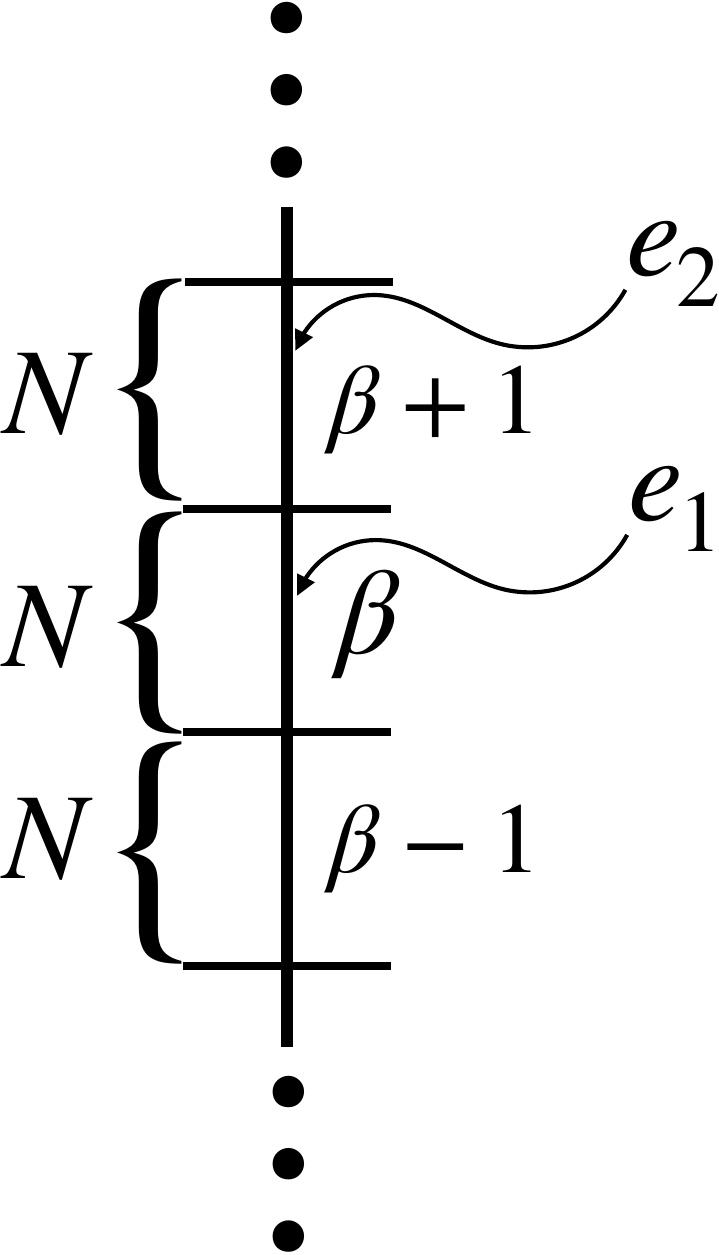}
\end{subfigure}
%    \caption{A recurrent graph.}
  %  \label{fig:example_graph}
  \caption{Left: The three numbers in a vertex $v$, left to right, are: its weight $w(v)$ (black), its potential $\Pot(v)$ ({\color{red} red}), and its strength $\St(v)$ ({\color{cyan} cyan}) (see App.~\ref{app:strengths} for details on the latter). Red edges depict \Max's choices upon winning a bidding; he proceeds to a neighbor with the maximal potential. 
Right: Energy blocks and virtual energies that belong to them.} 
\label{fig:block}
\end{figure}    

\paragraph*{The normalization scheme}
We describe the normalization scheme of the block strategy, which is key in our analysis.

The strategy chooses $N \in \Nat$ and $z>1$. The precise choice of the parameters $N$, $z$, and $e_0$ is insignificant for the analysis, we do not elaborate on it and refer the reader to~\cite{AHC19} for details. Our only assumption is that $N$ is larger than the sum of weights in any vine. 

We partition $\Q$ into {\em energy blocks} of size $N$ (see Fig.~\ref{fig:block}, Right), and the normalization factor is chosen based on the block to which the virtual energy resides in. Formally,  for $i \in \Nat$, the $i$-th energy block is $[Ni, N(i+1)) \cap \bbQ$. Note that the end points of the blocks are integers but the virtual energy can be rational. 
Denote by $\beta$ the energy block that $e(\prefN)$ belongs to, thus $\beta = \lfloor e(\prefN) / N \rfloor$. We call $\beta$ the {\em current energy block}. 

\smallskip
\noindent{\bf Normalization:} When the game is in vertex $v$ and the current energy block is $\beta$, the bid is $\St(v) \cdot z^{-\beta}$. Note that changes to the normalization factor occur only at the root. 
\smallskip

%SHAULL_COMFORTING
We emphasize that this exposition does not entirely clarify how block strategies behave, but merely gives the necessary information to carry on with our analysis.

\stam{
For $z > 1$, define a new game $\G^z = \zug{V, E, w^z}$ with 
% \[ 
% w^z(v) = \begin{cases}
%     w(v) & \text{if } w(v) \geq 0 \\
%     z \cdot w(v) & \text{if } w(v) < 0.
% \end{cases}
% \]
$w^z(v) = w(v)$ if $ w(v) \geq 0$, and  $w^z(v)= z \cdot w(v)$ if $w(v)<0$. 
Denote by $\Pot^z$ and $\St^z$ the potential and strengths in $\G^z$, respectively. 
}

\stam{
The block strategy proceeds as follows. 
\begin{itemize}[itemsep=0pt]
\item \underline{Pre-processing:} 
\begin{itemize}[itemsep=0pt]
\item Find the value $\MP\big(\RT(\G)\big)$, choose $\epsilon>0$, and decrease $\MP\big(\RT(\G)\big)-\epsilon$ from all weights. 
\item Choose $z > 1$ s.t. $\MP\big(RT(\G^z)\big) \geq 0$; find $\Pot^z$ and $\St^z$.
\item Choose an initial energy level $e_0$ based on the
initial budget.
\end{itemize}
\item {\bf Choice of normalization.} 
For $n \in \Mil$, find the current energy block $\beta$. Then, the normalization factor for the next vine 
is $z^{-\beta}$. 
\item {\bf Actions.} 
At $v \in V$, bid $\St(v) \cdot z^{-\beta}$ and proceed to a neighbor $v^+ \in \neig(v)$ such that $\Pot^z(u) \leq \Pot^z(v^+)$, for every $u \in \neig(v)$.
\end{itemize}

}

\begin{example}
%yak yak We illustrate how the normalization is chosen. This is key to our analysis. For completeness, we also illustrate the intuition behind the definition, and the details are orthogonal to this work. 
Ex.~\ref{ex:strengths} in App.~\ref{app:strengths}, demonstrates that $\block$ guarantees that whenever the energy decreases, \Max ``gains'' budget, and whenever the energy increases, \Max ``invests'' budget. 
Suppose that the virtual energy is $e(\prefN) = e_1$, at some turn $n \in \Mil$ (see Fig.~\ref{fig:block}, Right). Then, in the next vine, \Max chooses $z^{-\beta}$ as the normalization factor. 
Consider the case that \Min bids $0$ for several vines, thus \Max draws the game to $\ell_3$, and importantly for this example, the energy increases. Eventually, $e(\pref{n'}) = e_2$, at some turn $n' \in \Mil$. Then, \Max's normalization factor decreases to $z^{-(\beta+1)}$. If \Min now wins and repeatedly draws the game to $\ell_2$, the energy will eventually reach $e_1$, and the normalization increases to $z^{-\beta}$. 
%The choice of normalization implies an invariant: whenever the energy is $e$, \Max's budget exceeds $B(e)$, where $B(\cdot)$ has two properties. First, $B(e)$ tends to $0$ as $e$ increases, which means that with any initial budget $B_0 > 0$, \Max can choose a high initial energy $e_0$ so that the invariant holds initially, i.e., $B_0 > B(e_0)$. Second, $B(0) \geq 1$, which means that the energy will never drop to $0$. Indeed, otherwise, \Max's budget exceeds the total budget. It is not hard to show that this second property means that any play that is consistent with $\block$ has a non-negative mean-payoff. 
\end{example}

%The block strategy guarantees that no matter how Min plays, the virtual energy never drops to $0$, i.e.,  $e_0 + \energy(\prefN) \geq 0$, for every prefix $\prefN$.

%We describe the idea and refer the reader to~\cite{AHC19} for more details. The strategy maintains an invariant $I$ between the virtual energy and the budget of Max; namely, when the virtual energy is
% \shtodo{at most?}
%  $e$, then Max's budget is greater than $I(e)$. The choice of $e_0$ is such that the invariant holds initially. Virtual energy $0$ is never reached since $I$ is defined so that $I(0) > 1$, that is virtual energy $0$ means that Max's budget exceeds $1$, which is impossible. 

\paragraph*{The generated play.}
Consider a non-zero-sum recurrent game $\G = \zug{V, E, w_1, w_2}$ and an initial configuration $c_0 = \zug{v_0, B_1}$. That is, the token is initially placed on the root, \PO is allocated $B_1 > 0$ and \PT is allocated $B_2 = 1- B_1 >0$. For $i\in \{1,2\}$, let $\G_i = \zug{V, E, w_i}$ be a zero-sum game and construct the block strategy $\block_i$ for initial budget $B_i$ as above. For ease of presentation, we assume that $\block_1$ and $\block_2$ choose the same parameters $N$ and $z$, and in Rem.~\ref{rem:diff-params}, we discuss the changes required to lift this assumption. 
In the remainder of this section, we analyze the {\em generated} play $\play(c_0, \block_1, \block_2)$.

%%%%%%%%%%%%%%%%%%%%%%%%%%%%%%%%%%%%%%%%%%%%
%\subsection{The generated play forms an ultimately periodic path}
\subsection{The generated path is ultimately periodic}
In this section, we show that the path in $\G$ that corresponds to the play generated by block strategies is ultimately periodic.

\begin{figure}[ht]
\centering
\includegraphics[height=4cm]{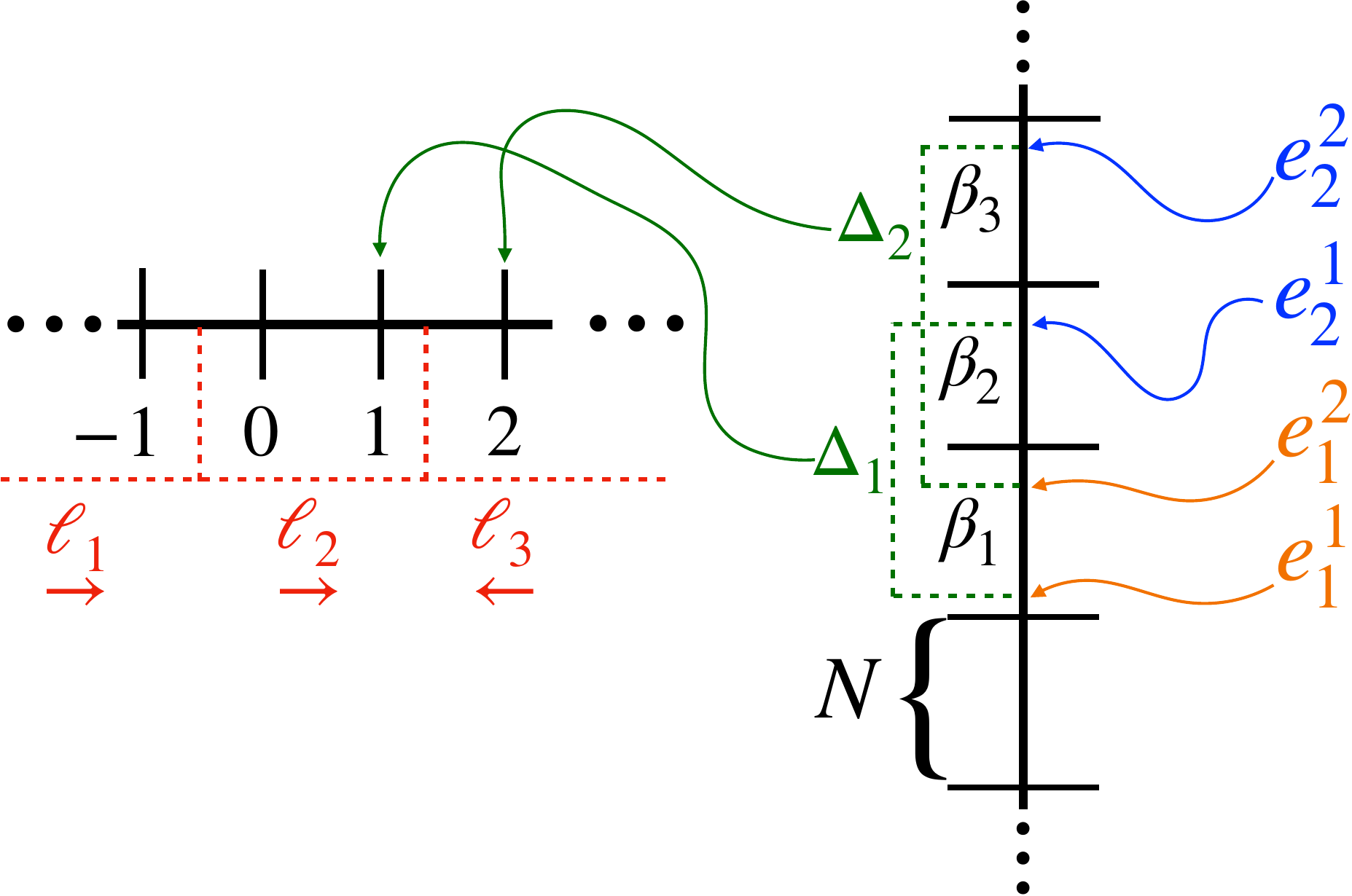}
\caption{Right: four virtual energies and the energy blocks that they belong to. Left: the partition of $\Q$ to difference intervals.}
\label{fig:Sec3}
\end{figure}

\paragraph*{Proof overview.}
We describe the main ingredients of the proof.

%We introduce some notation. For $i \in \set{1,2}$, denote by $\St^z_i$ the strengths in $\G^z_i$, by Let $e^0_i$ be the initial energy that $\block_i$ chooses, let $e_i(\prefN) = e^0_i + \sum_{0 \leq j <n} w_i(v_j)$ denote the virtual energy following prefix $\prefN$, thus its energy block is $\beta_i = \lfloor e_i(\prefN)/N \rfloor$. Recall that $\Mil$ are the milestones in which the generated play visits the root.

A central quantity is the difference between the energy blocks in a prefix, which we use $\Delta$ to denote. To illustrate, Fig.~\ref{fig:Sec3} Right, depicts four virtual energies. For $n_1, n_2 \in \Mil$ and $i \in \set{1,2}$, at at turn $n_i$, denote the virtual energy of the first coordinate by $e^i_1 = e_1(\pref{n_i})$ (depicted in \orange) and the second coordinate by $e^i_2 = e_2(\pref{n_i})$ (depicted in \blue). The energy-block difference $\Delta_1 = \beta_2-\beta_1 =1$ corresponds to $\pref{n_1}$ and $\Delta_2 = \beta_3-\beta_1 = 2$ corresponds to $\pref{n_2}$. 

We show that the energy block difference, $\Delta$, uniquely determines the next vine (Lem.~\ref{lem:interval}). The proof idea is that the normalization factor, which is tied to the difference, is determined at the root and the other aspects of the strategy (the strengths) are constant. 

We partition $\Q$ into intervals, and associate each interval with a leaf. When the token is at the root and energy-block difference is in an interval that corresponds to leaf $\ell$, then $\ell$ is the next leaf that the token visits. To illustrate, Fig.~\ref{fig:Sec3} Left, depicts three intervals separated by dashed red lines. For example, the energy-block differences $0$ and $1$ correspond to the middle interval, which is associated with leaf $\ell_2$, thus following prefix $\pref{n_1}$ with energy-block difference $\Delta_1$, the next leaf to be visited is $\ell_2$. 

We think of the sequence of energy-block difference traversed by the play as a walk on $\Z$. We associate with each interval a {\em direction}, which marks the direction in which the walk ``tends'' to proceed to (Def.~\ref{def:interval-dir}). For example, the directions in Fig.~\ref{fig:Sec3} are depicted beneath each interval. The middle interval corresponds to leaf $\ell_2$, which  might have $w_1(\ell_2)<w_2(\ell_2)$, meaning that the energy in the second coordinate (\blue) increases faster than the first coordinate (\orange), which causes the energy-block difference to ``tend'' to increase after visiting $\ell_2$. This is the situation depicted in the right part of the figure: following prefix $\pref{n_1}$, the energy-block difference is $\Delta_1$ and it grows to $\Delta_2$ in $\pref{n_2}$. 

Finally, a key lemma (Lem.~\ref{lem:non-fluc}) shows that the walk does not fluctuate arbitrarily; that is, it either tends to $\infty$, $-\infty$, or it is bounded. In the unbounded case, the walk eventually stays in one of the outer intervals. Staying in one interval means that the same leaf is visited repeatedly. In the bounded case, we identify an equivalence relation between pairs of virtual energies with finitely many classes, and deduce that eventually a cycle of vines is formed (Lem.~\ref{lem:shift-inv}).

\subsubsection{Determining the next vine}
We define the {\em energy-block difference} formally.

\begin{definition}%{\bf (Energy-block difference).}
For $n \in \Mil$, the {\em energy-block difference} is $\Delta(\prefN) =  \beta_2(\prefN)-\beta_1(\prefN)$. 
%The {\em energy-block difference sequence} is $\Delta_1,\Delta_2,\ldots$% \in \Z^\omega$ with $\Delta_j = \Delta(\pi^{\leq n_j})$, for $n_j \in \Mil$. 
\end{definition}

%Intuitively (see Fig.~\ref{fig:Sec3}), we associate with each leaf $\ell$ in $\G$ an interval $l(\ell) \subseteq \Q$ such that if $\Delta \in l(\ell)$, the next vine to be played ends in $\ell$. For example, $e^1_2$ belongs to the $\beta_2$-th block, $e^1_1$ belongs to the $\beta_1$-th block, and their difference $\Delta_1=1$ belongs to the interval $l(\ell_2)$, thus the next vine to be played ends in $\ell_2$. 

Next, we partition $\Q$ according to the energy-block differences (Fig.~\ref{fig:Sec3}, Left).

\stam{
\begin{figure}[ht]
\caption{Right: two virtual energies and the energy blocks that they belong to. Left: block numbers and $\Delta$. Bottom: the partition of $\Q$ to intervals.}
\label{fig:Sec3}
% Removed dead figure reference (inside \\stam, which discards its argument) for arXiv packaging.
\end{figure}
}

\begin{definition}
	\label{def:intervals}{\bf (Difference interval).}
The definition is inductive on the depth. For the root, define $l(v_0) = (-\infty, \infty)$.
Next, suppose that $l(v)$ is defined. For $i \in \set{1,2}$, let $u_i \in \neig(v)$ be the choice of $\block_i$ upon winning the bidding at $v$. Let $u \in \neig(v)$. If $\St^z_1(v), \St^z_2(v) \neq 0$,  define the {\em balance} of $v$  as  $\Bal(v) = \log_z\big(\frac{\St^z_2(v)}{\St^z_1(v)}\big)$ and define $l(u)$ as
    \[
        l(u) = \begin{cases}
            l(v))\cap (-\infty, \Bal(v))      & \text{if } u \neq u_1 \text{ and } u = u_2    \\
            l(v) \cap [\Bal(v), \infty)      & \text{if } u = u_1    \text{ and } u \neq u_2 \\
            l(v)                             & \text{if  } u = u_1   \text{ and } u = u_2    \\
            \emptyset                        & \text{otherwise} \\
        \end{cases}
    \]
If $\St^z_1(v)=0$, we define $l(u_2) = l(v)$. Otherwise $\St^z_2(v) =0$, 
and we define $l(u_1) = l(v)$. For every other $u \in \neig(v)$, define $l(u) = \emptyset$. 
%If either $\St^z_i(v) \neq 0$, for $i \in \set{1,2}$. Note that $\St^z_i(v) = 0$ means that $\block_i$ bids $0$ at $v$. If $\St^z_2(v)=0$, then $\block_1$ wins the bidding, and we define $l(u_1) = l(v)$. Otherwise $\St^z_2(v) =0$, and we define $l(u_2) = l(v)$. For every other $u \in \neig(v)$, define $l(u) = \emptyset$. 
\stam{%short
    \[
        l(u) = \begin{cases}
            l(v)                             & \text{if } u = u_1    \text{ and } \St_2(v) = 0 \\
            l(v)                             & \text{if } u = u_2    \text{ and } \St_1(v) = 0\\
            \emptyset                        & \text{otherwise} \\
        \end{cases}
    \]
    }
\end{definition}

Note that the boundaries of the intervals can be irrational numbers whereas we are interested in integer points within the intervals. The lemma below shows how the next vine is determined.

%The proof is in App.~\ref{app:proof of lem:interval}. In a nutshell, the definition of $\Bal(v)$ is tailored to account for the budget of the players such that the lemma is ensured.
\begin{lemma}
\label{lem:interval}
Let $n \in \Mil$ and a vertex $v \in V$. If $\Delta(\prefN) \in l(v)$, then the next vine visits $v$.
\end{lemma}
\begin{proof}
We prove by induction on the depth of $v$. The base case is trivial since every vine visits $v_0$. 
 For the inductive step, let $v'$ be the parent of $v$. 
Since $l(v) \subseteq l(v')$, we have $\Delta \in l(v')$, and by the induction hypothesis, the vine visits $v'$. 
Consider the bidding at $v'$. For $i \in \set{1,2}$, \PLi bids $b_i = \St^z_i(v_k) \cdot z^{-\beta_i}$ and proceeds to $u_i \in \neig(v)$. 
Suppose that $\St^z_1(v_k), \St^z_2(v_k) \neq 0$ and the other case is similar. \PO wins the bidding iff $b_1 \geq b_2$ iff $1 \geq b_2/b_1 = (\St^z_2(v) \cdot z^{-\beta_2})/(\St^z_1(v) \cdot z^{-\beta_1})$. Rearranging, \PO wins the bidding iff $\Bal(v) \leq \Delta$. If $u_1 \neq u_2$, note that $l(u_1)$ and $l(u_2)$ are disjoint and for $i \in \set{1,2}$, \PLi wins the bidding iff $\Delta \in l(u_i)$, in which case the token moves to $u_i = v$. If $u_1 = u_2$, then $l(v) = l(v')$, and the game proceeds to $v$ no matter the bidding outcome. 
\end{proof}

%%%%%%%%%%%%%%%%%%%%%%%%%%%%%%
\subsubsection{The energy-block difference walk}
\label{sec:walk}
In this section, we prove that the energy-block difference walk does not fluctuate arbitrarily. Formally, consider the sequence $\set{e_2(\prefN) - e_1(\prefN)}_{n \in \Mil}$, then the walk is $\Delta_1,\Delta_2,\ldots$ with $\Delta_j = \Delta(\pi^{\leq n_j})$, for $n_j \in \Mil$. The following definition intuitively guides the direction of the walk.

\begin{definition}{\bf (Interval directions).} 
\label{def:interval-dir}
Let $\L$ be the set of leaves in $\G$. For $\ell\in \mathcal{L}$, let $v_0,\ldots, v_k =\ell$ be the unique path from the root to $\ell$, and define $W_i(\ell)=\sum_{0 \leq j \leq k}w_i(v_j)$ be the sum of weights on the path.  
Define a labeling $\delta$ of the leaves where $\delta(\ell)$ is $\rightarrow/\leftarrow/\bot$ if $W_2(\ell)-W_1(\ell)$ is positive/negative/zero, respectively.
% \[
%     \delta(\ell) = \begin{cases}
%         \rightarrow   & \text{if } W_2(\ell) - W_1(\ell) > 0 \\
%         \leftarrow    & \text{if } W_2(\ell)-W_1(\ell) <  0\\
%         \bot          & \text{if } W_2(\ell)-W_1(\ell) =  0
%     \end{cases} 
% \]
\end{definition}
Intuitively, 
%as we have seen (Lem.~\ref{lem:interval}), the energy-block difference sequence $\Delta_1,\Delta_2,\ldots$ determines the sequence of vines that are played. We think of the sequence as a walk on $\Z$. If 
if for $j \geq 1$, the $j$-th vine ends in a {\em right-leaning} interval, i.e., $\ell \in \L$ having $\delta(\ell) = \rightarrow$, the virtual-energy difference increases, thus we expect the walk to proceed ``right'', i.e., $\Delta_{j+1} \geq \Delta_j$. This intuition is useful but not precise. We illustrate why. Suppose that $e_1(\pref{j})$ is close to the boundary of the $\beta_1$-th block whereas $e_2(\pref{j})$ is far from the boundary of the $\beta_2$-th block. Then, for $W_1(\ell),W_2(\ell)> 0$, it is possible that only one of the virtual energies changes blocks: $e_1(\pref{j+1})$ to the $\beta_1+1$ block while $e_2(\pref{j+1})$ stays in the $\beta_2$ block, which results in $\Delta_{j+1} < \Delta_j$. 
%\shtodo{Maybe we can delegate the whole $\Gamma$ thing to the appendix: remove this and Lemma 2+Lemma 3, just keep Lemma 4. Then we also need to remove the Proof of Theorem 3.}
%Nonetheless, %in App.~\ref{app:step size} 
We show the following lemma, which intuitively shows that when the walk enters an interval, the first index in the interval is not far from the boundary. 

\begin{lemma}
\label{lem:step-size}
For $j \in \Mil$ and its successor $j' \in \Mil$, we have $|\Delta(\pref{j}) - \Delta(\pref{j'})| \leq 2$. 
\end{lemma}
\begin{proof}
We reason about another quantity: for $n \in \Mil$, define $\Gamma(\prefN) = \floor{\frac{e_2(\prefN) - e_1(\prefN)}{N}}$. The following claim connects $\Gamma$ and $\Delta$. Its proof is technical and we omit it. 

\begin{claim}
\label{cl:floors}
For all $x,y \in \Q$, the following holds: 
     \[
         \lfloor y- x \rfloor - 1 \leq  \floor{y} -  \floor{x} \leq \floor{ y - x} + 1 
     \]
 Therefore, for $n \in \Mil$, we have $\Gamma(\prefN) -1 \leq \Delta(\prefN) \leq \Gamma(\prefN) + 1$. 
\end{claim}

The following claim 
%lemma (see the proof in App.~\ref{app:invariant}) 
intuitively shows that updates to $\Gamma$ agree with our intuition: in a right-leaning interval, $\Gamma$ grows, and in a left-leaning interval, it shrinks. The lemma requires that $\Gamma$ is far from the interval's boundaries, guaranteeing that $\Delta$ belongs to the same interval.
 \begin{claim}
\label{cl:invariant}
 Consider an interval $l(\ell)$, for some $\ell \in \L$. Let $a \in \Z$ such that $\set{a-1, a, a+1} \subseteq l(\ell)$. For $j \in \Mil$ and its successor $j' \in \Mil$ we have:
     \begin{enumerate}
         \item if $\delta(\ell) = (\rightarrow)$ and $\Gamma(\pref{j}) = a$, then $\Gamma(\pref{j'}) \geq a$, and
         \item if $\delta(\ell) = (\leftarrow)$ and $\Gamma(\pref{j}) = a$, then $\Gamma(\pref{j'}) \leq a$.
     \end{enumerate}
 \end{claim}
%\begin{proof}
\noindent{\em Proof of claim:}
    We prove the first part, and the proof for the other case is dual. 
Claim.~\ref{cl:floors} implies that $\Delta(\pref{j}) \in l(\ell)$. By Lem.~\ref{lem:interval}, the next milestone ends in $\ell$. Then,
\[
\begin{split}
&\lfloor \frac{e_2(\pref{j'}) - e_1(\pref{j'})}{N} \rfloor \\
&=  \lfloor \frac{e_2(\pref{j}) - e_1(\pref{j}) + W_2(\ell) - W_1(\ell)}{N} \rfloor \\
&\geq \lfloor \frac{e_2(\pref{j}) - e_1(\pref{j})}{N} \rfloor = a.
\end{split}
\]
The inequality follows from $\delta(\ell) = (\rightarrow)$.\hfill\qed (of claim)
%\end{proof}

%Intuitively, suppose that the walk enters an interval, then the next lemma shows that the first index in the interval is not far from the boundary. 
The proof of Lemma~\ref{lem:step-size} now follows from our assumption that $N$ is larger than the sum of weights in a path in $\G$. Thus, if the energy block shifts, it shifts to a neighboring block, implying an upper bound on a step in the walk. 
\end{proof}

We conclude this section by showing that the virtual-energy difference sequence does not fluctuate unboundedly. % (see App.~\ref{app:proof of lem:non-fluc}). 

\begin{lemma}
\label{lem:non-fluc}
The sequence $\set{e_2(\prefN) - e_1(\prefN)}_{n \in \Mil}$ either tends to $\infty$, to $-\infty$, or is bounded. 
\end{lemma}
 \begin{proof}
 We prove for the lemma for the sequence $\set{\Gamma(\prefN)}_{n \in \Mil}$, which clearly implies the lemma. 
 Denote the rightmost interval by $l(\ell_R) = [a,\infty)$. We distinguish between the three possibilities of $\delta(\ell_R)$. 
 In the first case, $\delta(\ell_R) = \leftarrow$. We show that sequence is bounded from above; namely, we show that for all $n \in \Mil$, we have $\Gamma(\prefN) \leq \max\set{\Gamma(\pref{1}),a+2}$. 
 If $\Gamma$ is within the interval $l(\ell_R)$, it only decreases. Formally, for $j \in \Mil$ and its successor $j' \in \Mil$, by Claim.~\ref{cl:invariant}, if $\Gamma(\pref{j}) \geq a+2$, then $\Gamma(\pref{j'}) \leq \Gamma(\pref{j})$. It is left to show that cannot enter $l(\ell_R)$ in arbitrarily high indices. By Lem.~\ref{lem:step-size}, if $\Delta(\pref{j'}) \in l(\ell_R)$ and $\Delta(\pref{j}) \notin l(\ell_R)$, then $\Delta(\pref{j'}) \leq a+1$, and by Claim.~\ref{cl:floors}, we have $\Gamma(\pref{j'}) \leq a+2$. 

 For the other two cases, assume towards contradiction that the sequence fluctuates arbitrarily. Then, there exists  $j \in \Mil$ with $\Gamma(\prefN) \geq (a+2)$. 
 %A simple calculation shows $\Gamma(\pref{j}) = \floor{\frac{e_2(\pref{j}) - e_1(\pref{j})}{N}} \geq (a+1)$. 
 By Claim.~\ref{cl:floors}, we have $\Delta(\pref{j}) \in l(\ell_R)$ and by Lem.~\ref{lem:interval}, the next vine visits $\ell_R$. Let $j' \in \Mil$ be the successor of $j$. 
 If $\delta(\ell_R) = \bot$, then $\Gamma(\pref{j}) = \Gamma(\pref{j'})$,
% %e_2(\pref{j'}) - e_1(\pref{j'}) = e_2(\pref{j}) - e_1(\pref{j})$, 
 meaning that the sequence stays constant.
 If $\delta(\ell_R) = \rightarrow$, then by Claim.~\ref{cl:invariant}, $\Gamma(\pref{j'}) \geq a+2$, meaning that all next vines traverse $\ell_R$, and the sequence grows indefinitely. Both contradict the assumption.
 % in a constant step size of $W_1(\ell) - W_2(\ell)$, 
 %(OLD) If $\delta(\ell_R) = \leftarrow$, then by the assumption as the energy difference fluctuates we can assume that the predecessor of $j$, denoted $j''$, has that $e_2(\pref{j''}) - e_1(\pref{j''}) \leq N \cdot (a+1)$ and thus by the choice of $N$, $e_2(\pref{j}) - e_1(\pref{j}) \leq e_2(\pref{j''}) - e_1(\pref{j''}) + N \leq N \cdot (a+2)$. And by \ref{lem:invariant} $e_2(\pref{j'}) - e_1(\pref{j'}) \leq e_2(\pref{j}) - e_1(\pref{j}) \leq N \cdot (a+2)$, which holds invariant.
 \end{proof}

%%%%%%%%%%%%%%%%%%%%%%%%%%%%%%%%%
\subsubsection{The generated path is ultimately periodic}
The next lemma is used to reason about the case that $\set{e_2(\prefN) - e_1(\prefN)}_{n \in \Mil}$ is bounded. 
It shows that the absolute virtual energy does not contribute to determining the next vine to be played, rather the key factor is the difference between virtual energies. 
%\shtodo{added text and moved proof to the appendix:}
The proof 
%(see App.~\ref{app:proof of lem:shift-inv}) 
is a consequence of Lem.~\ref{lem:interval}.
%OLD two energies by $N \cdot k$, for $k \in Z$, does not affect the vine that is played the key factor that determines the next vine is the difference between the two virtual energies and not their  

\begin{lemma}\label{lem:shift-inv}{\bf (Energy-block shift invariance).}
Let $n,n' \in \Mil$ such that there exists $k \in \Nat$ for which $e_i(\prefN) = e_i(\pref{n'}) + k\cdot N$, for $i \in \set{1,2}$. Then, the vine that is played following $\prefN$ and $\pref{n'}$ is the same. It follows that the path in $\G$ that is traversed by the suffixes of $\prefN$ and $\pref{n'}$ is the same. 
\end{lemma}
\begin{proof}
Observe that:
\[
 \Delta(\prefN) = %\beta_2(\prefN) - \beta_1(\prefN) = 
 \lfloor \frac{e_2(\prefN)}{N} \rfloor - \lfloor \frac{e_1(\prefN)}{N} \rfloor =\] \[= \lfloor \frac{e_2(\prefN) + kN}{N} \rfloor - \lfloor \frac{e_1(\prefN) + kN}{N} \rfloor 
 =\lfloor \frac{e_2(\pref{n'})}{N} \rfloor - \lfloor \frac{e_1(\pref{n'})}{N} \rfloor = \Delta(\pref{n'}).
 \]
 Thus, $\Delta(\prefN)$ and $\Delta(\pref{n'})$ belong to the same interval and by Lem.~\ref{lem:interval}, the same vine is played next.
 \end{proof}

%%%%%%%%%%%

%We turn to the main result of this section, showing that the path in $\G$ that corresponds to the generated play is ultimately periodic.
Our main result of this section is that 
the path in $\G$ that corresponds to the generated play is ultimately periodic.

\begin{theorem}
\label{thm:budget-lasso}
 For every initial configuration $c_0$ it holds that $\path(c_0, $ $ \block_1, \block_2) = \tau_1 \cdot \tau_2^\omega$, where $\tau_1,\tau_2 \in V^*$. Moreover, let $q \in \Q$ be the granularity of the weights, i.e., all weights in $\G$ are of the form $q \cdot m$, for $m \in \Z$. Then, the mean-payoff values of the generated play is computable in $O\big(\frac{N^2}{q^2}\cdot (\max_v \Bal(v) - \min_v \Bal(v))\big)$.
% \shtodo{Maybe move ``granularity'' outside the theorem.}
\end{theorem}
\begin{proof}
It follows from Lem.~\ref{lem:non-fluc} that there are three possible outcomes for the block-difference sequence $\Delta_1,\Delta_2,\ldots$. The first two are that it  eventually stays in one of the infinite intervals, in which case eventually, only one vine is played. 

We consider the case in which the sequence is bounded. 
Let $n \in \Mil$. We associate with $\prefN$ an {\em energy configuration} $\zug{d^\downarrow(\prefN), d^\updownarrow(\prefN)}$, where $d^\updownarrow(\prefN) = (e_2(\prefN)-e_1(\prefN))$ and $d^\downarrow(\prefN)$ is the distance from the lower virtual energy to the boundary of the energy block that it belong to, formally for $e = \min\set{e_1(\prefN), e_2(\prefN)}$, assume that $e$ belongs to the $\beta$-th energy block, then $d^\downarrow(\prefN) = (e - \beta \cdot N)$. 
Observe that for $n, n' \in \Mil$ with equality between the energy configurations, i.e., $\zug{d^\downarrow(\prefN), d^\updownarrow(\prefN)} = \zug{d^\downarrow(\pref{n'}), d^\updownarrow(\pref{n'})}$, the pairs of energies are block shifted, i.e., there exists $k \in \Nat$ such that $e_i(\prefN) = e_i(\pref{n'}) + k N$, for $i \in \set{1,2}$. Thus, by Lem.~\ref{lem:shift-inv}, the next vine that is played from both prefixes is the same. 
Finally, we count the number of energy configurations that a play with a bounded energy-block difference sequence traverses. Since the length of an energy block is $N$,
% \shtodo{here we explicitly assume that the energy block is a subset of $\bbN$, no?}
% \jutodo{No, this is the energies that are rationals.}
% \shtodo{So why is its size $N$? Seems like its size is infinite.}
% \jutodo{Each block is $[Ni, Ni + N)$, so the size is $N$, and since the granularity is $q$, there are $N/q$ possible values for $d^\downarrow$.}
% \shtodo{I see it's changed to ``length''. I guess it's ok (although usually intervals of rationals are not associated with length).}
 there are $N/q$ possibilities for $d^\downarrow$. As seen in Lem.~\ref{lem:non-fluc}, the maximal and minimal possible values of $\Delta$ are respectively $\max_v \Bal(v) + 1$ and $\min_v \Bal(v) - 1$, thus there are at most $(\max_v \Bal(v) - \min_v \Bal(v)) + 2)\cdot \frac{N}{q}$ possible values for $d^\updownarrow$. Since both are finite, eventually an energy configuration repeats, from which point the same sequence of vines is traversed.

%In App~\ref{app:alg budget lasso} we show how to obtain from the proof an algorithm to compute the payoffs of the generated play.

We conclude with 
an algorithm to compute the payoffs of the generated play. Construct a graph in which each vertex corresponds to an energy configuration. 
As mentioned above, each $d = \zug{d^\downarrow, d^\updownarrow}$ determines the next vine to be played $M(d)$. The (unique) neighbor of $d$ is the energy configuration that corresponds to the energy updates following $M(d)$. More formally, suppose that a pair of energies $\zug{e_1, e_2} \in \Q^2$ is mapped to $d$, then the neighbor of $d$ is the energy configuration that $\zug{e_1 + W_1\big(M(d)\big), e_2 + W_2\big(M(d)\big)}$ is mapped to. We add two sinks that corresponds to unbounded walks. Once the graph is constructed, all that is left is to find the vertex $d_0$ that corresponds to the initial choice of energies given by the initial configuration $c_0$, and find the ultimately period path in the graph that starts from $d_0$. Let the path be $\tau_1 \cdot \tau_2$ with $\tau_2 = d_1,\ldots, d_k$. Then, for $i \in \set{1,2}$, we have $\MP_i(c_0, \block_1, \block_2) = \sum_{1 \leq j \leq k} W_i\big(M(d_j)\big) / \sum_{1 \leq j \leq k} |M(d_j)|$.
% We conclude with an algorithm to compute the payoffs of the generated play. Construct a graph in which each vertex corresponds to an energy configuration. 
% As mentioned above, each $d = \zug{d^\downarrow, d^\updownarrow}$ determines the next vine to be played $M(d)$. The (unique) neighbor of $d$ is the energy configuration that corresponds to the energy updates following $M(d)$. More formally, suppose that a pair of energies $\zug{e_1, e_2} \in \Q^2$ is mapped to $d$, then the neighbor of $d$ is the energy configuration that $\zug{e_1 + W_1\big(M(d)\big), e_2 + W_2\big(M(d)\big)}$ is mapped to. We add two sinks that corresponds to unbounded walks. Once the graph is constructed, all that is left is to find the vertex $d_0$ that corresponds to the initial choice of energies given by the initial configuration $c_0$, and find the ultimately period path in the graph that starts from $d_0$. Let the path be $\tau_1 \cdot \tau_2$ with $\tau_2 = d_1,\ldots, d_k$. Then, for $i \in \set{1,2}$, we have $\MP_i(c_0, \block_1, \block_2) = \sum_{1 \leq j \leq k} W_i\big(M(d_j)\big) / \sum_{1 \leq j \leq k} |M(d_j)|$.
\end{proof}

We point to an interesting corollary of Thm.~\ref{thm:budget-lasso}. Let $C$ be the cycle of vines that the generated play converges to. Observe that if the energy difference is bounded, then $\energy_1(C) = \energy_2(C)$ and \PLi's payoff is $\MP_i(\pi) = \energy_i(C)/|C|$, for $i \in \set{1,2}$. 
%We remind the suspicious reader that when constructing $\block_i$, we first normalize the weights so that $\MP_1(\G) = \MP_2(\G) = \epsilon$. 

\begin{corollary}{\bf (Fair payoffs).}\label{cor:fair}
Consider a game $\G$ such that $e_2(\prefN) - e_1(\prefN)$ is bounded. Then, $MP_1(\pi) = MP_2(\pi)$. 
\end{corollary}

\begin{example}
Recall the two generated plays that are described in Ex.~\ref{ex:intro}. The walk when \blue plays against \orange is bounded, thus both payoffs are $0.5$, however when \blue plays against \red, the walk is unbounded and the payoffs differ. Fig.~\ref{fig:plots2} is similar to the plots in Fig.~\ref{fig:plots} only at a reduced granularity. This highlights the slopes of the plots, which coincide with the mean-payoff.
\begin{figure}[ht]
\centering
\includegraphics[width=4.2cm]{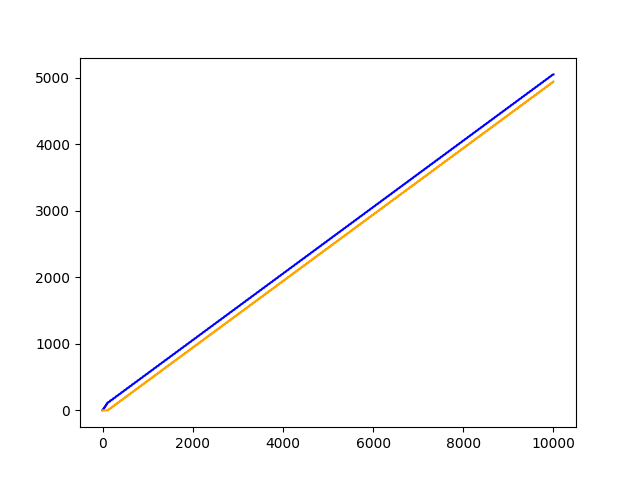}
\includegraphics[width=4.2cm]{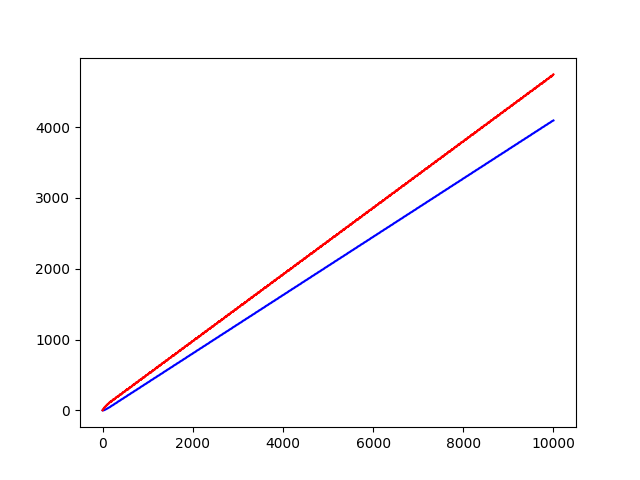}
\caption{The plots for generated plays as in Fig.~\ref{fig:plots} only that the simulation is for 10k turns rather than 500.}
\label{fig:plots2}
\end{figure}
\end{example}

\begin{remark}[Choosing different parameters]
\label{rem:diff-params}
We describe the changes needed to lift the assumption that the two block strategies choose the same parameters $N$ and $z$. 
First, Def.~\ref{def:intervals} is used to anticipate which vine will be played next based on the energy-block difference of the strategies. The choice of $z$ affects the normalization of the bids throughout the vine, and note that the normalization does not directly depend on $N$. Thus, when the two strategies choose different values of $z$, we need a more complicated expression in the definition: $\Bal(v) = (\log(z_2)/\log(z_1)) \cdot \log(St_2(v)/St_1(v))$.  
Second, different choices of $N$ mean that the upper bound on the number of block differences in Thm.~\ref{thm:budget-lasso} needs to grow. 
\end{remark}

%%%%%%%%%%%%%%%%%%%%%%%%%%%%%%%%%%%%
%\shtodo{Maybe turn this subsection to a paragraph? (remove the proofs and maybe shorten the phrasing (e.g., combine Assumption 1 and Theorem 3).}
\subsection{Games with long intervals have small cycles}
In this section we assume that the intervals are large:
% \begin{assumption}
 %\label{asmp:large-intr}
for every leaf $\ell$, we have $|l(\ell) \cap \Z| \geq 3$. 
 %\end{assumption}
We characterize the recurrent behavior of such games. Let $\inf(\pi)$ denote the set of leaves that are visited infinitely often in $\pi$. %(See \cref{app:two-vines} for the proof).
\begin{theorem}
\label{thm:two-vines}
If for every leaf $\ell$, we have $|l(\ell) \cap \Z| \geq 3$, then $|\inf(\pi)| \leq 2$.  
\end{theorem}
\begin{proof}
Assume towards contradiction that $\pi$ visits at least three vines infinitely often. Observe the walk $\Delta_1,\Delta_2,\ldots$. By Lem.~\ref{lem:interval}, the walk must traverse three adjacent intervals $l(\ell_1), l(\ell_2)$, and $l(\ell_3)$ infinitely often. 
Denote $l(\ell_2) = [a,b]$. Due to symmetry, assume Wlog that $\delta(l(\ell_2)) = (\rightarrow)$ or $\delta(l(\ell_2)) = \bot$. 
Let $i \in \Mil$ such that the walk visits $l(\ell_3)$, i.e.  $\Delta(\pref{i}) \geq a + 3$. By 
 the premise of the theorem, it holds that each leaf $\ell$ has $|l(\ell) \cap \Z| \geq 3$. 
%Asmp.~\ref{asmp:large-intr}, 
 $\Gamma(\pref{i}) \geq a + 1$. %App.~\ref{app:two-vines} 

 We prove that for every $k\geq i$, we have $\Gamma(\pref{k}) \geq a+1$. By Claim.~\ref{cl:floors},  $\Delta(\pref{k}) \geq a$. Thus, $l(\ell_1)$ is not reached, which is a contradiction. 
 Assume towards contradiction that $\pi$ visits at least three vines infinitely often. Observe the walk $\Delta_1,\Delta_2,\ldots$. By Lem.~\ref{lem:interval}, the walk must traverse three adjacent intervals $l(\ell_1), l(\ell_2)$, and $l(\ell_3)$ infinitely often. 
Denote $l(\ell_2) = [a,b]$. Due to symmetry, assume Wlog that $\delta(l(\ell_2)) = (\rightarrow)$ or $\delta(l(\ell_2)) = \bot$. 
Let $i \in \Mil$ such that 
%the walk visits $l(\ell_3)$, i.e. 
$\Delta(\pref{i}) \geq a + 3$. By 
the premise of the theorem, it holds that 
%Asmp.~\ref{asmp:large-intr}, 
$\Gamma(\pref{i}) \geq a + 1$. 

%App.~\ref{app:two-vines} proves 
It is now sufficient to prove that for every $k\geq i$, we have $\Gamma(\pref{k}) \geq a+1$. Indeed, by Claim.~\ref{cl:floors},  $\Delta(\pref{k}) \geq a$. Thus, $l(\ell_1)$ is not reached, which is a contradiction. 

We prove this holds. 
If $\Gamma(\pref{k}) \geq a + 1$ then $\Delta(\pref{k}) \geq a$, and thus either we visit $l(\ell_2)$ at iteration $k$ and so $\Gamma(\pref{k+1}) \geq \Gamma(\pref{k}) \geq a + 1$ by our assumption on $\delta(l(\ell_2))$, or we visit $l(\ell_3)$ at the current step, thus $\Gamma(\pref{k}) \geq \Delta(\pref{k}) - 1 \geq a + 2$ but notice that by the choice of $N$ we get that $\Gamma(\pref{k+1}) \geq \Gamma(\pref{k}) -1 \geq a + 1$.  In all cases we proved that for every $k \geq i$ it holds $\Gamma(\pref{k}) \geq a + 1$

% % \stam{

% %     In all cases we proved that $\Gamma(\pref{k}) \geq a_2 + 2$, and thus by Claim.~\ref{cl:floors} for every $k > i$, we have that  $\Delta(\pref{k}) \geq a_2$, which contradicts the assumption that $l_1\in \text{inf}(\pi)$.
% % Assume Wlog that $\delta(\ell_2) = \bot$ or $\delta(\ell_2) = \rightarrow$. 
% % Suppose that the walk is in $l(\ell_3)$. By Lem.~\ref{lem:step-size} and Assumption~\ref{asmp:large-intr}, it must visit $l(\ell_2)$ before visiting $l(\ell_1)$. Moreover, let $\Delta(\prefN)$ be the first visit to $l(\ell_2)$, then $\Delta(\prefN) \geq (b-1) \geq (a+2)$. 
% % \shtodo{what are $a$ and $b$? (I'm guessing the end points of the interval, but we should specify.)}
% % A similar argument as in the proof of Lem.~\ref{lem:non-fluc}, shows that $\Delta(\prefN) \geq (a+1)$ invariantly, contradicting that $l(\ell_1)$ is reached.
% % }
 \end{proof}

We characterize the payoffs below. %(see App.~\ref{app:short-MP}). 
For a leaf $\ell$, let $\len(\ell)$ be the length of the vine from $v_0$ to $\ell$. 
\begin{theorem}
\label{thm:short-MP}
Let $i \in \set{1,2}$. If $\inf(\pi) = \{\ell\}$,
then $\MP_i(\pi) = W_i(\ell)/\len(\ell)$. If $\inf(\pi) = \set{\ell_1,\ell_2}$, then $\MP_i(\pi) = \frac{\alpha W_i(\ell_1) + W_i(\ell_2)}{\alpha \len(\ell_1) + \len(\ell_2)}$, where $\alpha = \frac{W_1(\ell_1) - W_2(\ell_1)}{W_2(\ell_2) - W_1(\ell_2)}$.
\end{theorem}
\begin{proof}
The case of one vine is easy to see. Suppose that two vines $\ell_1$ and $\ell_2$ are visited infinitely often. Observe that this implies that the energy difference is bounded, thus denoting by $C$ the period of $\pi$, we have $\energy_1(C) = \energy_2(C)$. Suppose that the number of times that $C$ traverses $\ell_1$ and $\ell_2$ is $k$ and $t$, respectively. Thus, $\energy_i(C) = k \cdot W_i(\ell_1) + t \cdot W_i(\ell_2)$, for $i \in \set{1,2}$. Using $\alpha$ as in the statement leads to $k =  \alpha \cdot t$ and to the mean-payoff expression. 

Now suppose that $\text{inf}(\pi) = \{l(v_1), l(v_2)\}$, and assume that the cycle visits $v_1$ and $v_2$ $k$ and $t$ times respectively, then we enter a cycle $C$ where $W_1(C) = W_2(C)$ such that $W_i(C)$ is the added energy to player $i$ through the cycle, since the energy difference is bounded. We can see that,
  \[
      \forall i \in \{1,2\} w^i(C) = kw^i(P_1) + tw^i(P_2)
  \]
 and by the equality, denoting $\alpha = \frac{W_2(v_2) - W_1(v_2)}{W_1(v_1) - W_2(v_1)}$ we get that $k =  \alpha \cdot t$. 
    
It is easy to see that every visit to $v_j$ adds $W_i(v_j)$ to the energy of player $i$, and takes $\lambda(v_j)$ steps. Since the prefix before the cycle and the suffix of a partial cycle are both finite, we can ignore them and compute the mean-payoff of the game as:
\[
      \MP_i(\pi) = \frac{W_i(C)}{\lambda(C)} = \frac{kW_i(v_1) + tW_i(v_2)}{\lambda(v_1) k + \lambda(v_2) t}
  \]
  \[
      = \frac{\alpha t W_i(v_1) + tW_i(v_2)}{(\lambda(v_1) \alpha + \lambda(v_2)) t} =\frac{\alpha W_i(v_1) + W_i(v_2)}{\lambda(v_1) \alpha + \lambda(v_2)}
  \]
\end{proof}

\begin{example}
Consider the game depicted in Fig.~\ref{fig:example}. Note since two leaves are reachable, there are two non-empty intervals, both of which are unbounded, thus %Assumption~\ref{asmp:large-intr} holds. 
Thm.~\ref{thm:two-vines} applies.
Plugging $W_1(\ell_1) = 1$, $W_2(\ell_1) = 2$, $W_1(\ell_3) = 2$, $W_2(\ell_3) = 0$, $\len(\ell_1) = 3$, and $\len(\ell_3)=2$, thus $\alpha = 2$ and $\MP_1(\pi) = \MP_2(\pi) =\frac{2 \cdot 2 + 0}{2 \cdot 3 + 2} = 0.5$.
\end{example}

%%%%%%%%%%%%%%%%%%%%%%%%%%%%%%%%%%%%%%%%%%

\section{The Generated Path for Budget Strategies}
We now turn to analyze the path induced by \emph{budget strategies}, introduced in~\cite{AJZ21}. Intuitively, in a budget strategy each player bids a constant fraction of their budget, based on the current vertex, and chooses the successor with the highest potential, if they win. 

We show that for recurrent games, such strategies yield a piecewise-linear dynamics of the budgets. We give a condition on these dynamics that guarantees the path that corresponds to the strategies is ultimately periodic (c.f. Thm.~\ref{thm:disc finite then lasso}). However, showing this condition holds in the general case is currently out of reach. Nonetheless, we manage to show it holds for the subclass of recurrent games whose tree is of depth $1$, called {\em repeated bidding games}. 
%\emph{bowtie games}.

\paragraph*{The budget strategy}
We describe the budget strategy $\budget_i$ for \PLi, for $i \in \set{1,2}$, as constructed in~\cite{AJZ21}. The main difference from the block strategy is that the choice of normalization factor now depends on the current budget. The strategy proceeds as follows. 
%For $\varepsilon > 0$, we solve $\RT(\G_i, \frac{1}{2+\varepsilon})$, i.e., the random-turn game with a coin with bias $\frac{1}{2+\varepsilon}$ towards Max, to get potentials and strengths. 
\PLi chooses a constant $\alpha_i\in (0,1)$ such that $(1 - \alpha_i) = (1 + \alpha_i)^{-(1+\varepsilon)}$, which is shown to always exist. 
When reaching vertex $v$ with budget $B_i$, \PLi bids $b_i = \alpha_i \frac{\St_i(v)}{\St_{max}}\cdot B_i$, where $\St_{max}$ is the maximal strength among \PLi's strengths, and upon winning the bidding, proceeds to $v^+$ that maximizes the potential similar to the block strategy (see App.~\ref{app:strengths}). 
%SHAULL_COMFORT
As in the case of block strategies, we remark that the presentation of the budget strategies is condensed and does not clarify the entirety of the construction. We only include the bare minimum required to understand our analysis.

\subsection{A semi-algorithm to find payoffs in general recurrent games}
Consider a recurrent game $\G$ with initial vertex $v_0$, and fix budget strategies $\budget_1,\budget_2$ for the players. 
Given an initial budget $c$ for \PO, the strategies uniquely determine the next vine $\theta(c)$.
%a vine $\theta(c)$ taken from $v_0$ until the next milestone. 
\begin{definition}
\label{def:budget update function}
The \emph{budget-update function} is $\xb:[0,1]\to [0,1]$ such that $\xb(c)$ is the budget at the next milestone when starting at $v_0$ with budget $c$.
%We denote the budget in the next milestone as $\xb(c)$, and refer to $\xb:[0,1]\to [0,1]$ as the 
\end{definition}

Observe that from an initial budget $c_0$, the sequence of budgets $c_0,c_1,\ldots $ in the milestones of $\play(c_0,\budget_1,\budget_2)$ satisfies $c_i=\xb^i(c_0)$, where $\xb^i$ is the composition of $\xb$ with itself $i$ times. Thus, the entire behavior of the play is determined by the dynamics induced by $\xb$. 
In the remainder of this section we study this dynamics. The first step is to show that $\xb$ is piecewise-linear. 
\begin{example}
	\label{xmp:bowtie piecewise linear}
	\stam{
	Recall that a \emph{bowtie game} is a recurrent game of depth $1$. Since budget strategies always choose the same neighbor upon winning a bid, we can assume that a bowtie game has vertices $\{v_0,v_1,v_2\}$ where $v_0$ is the root and $v_1$ (resp. $v_2$) is the vertex chosen if \PO (resp. \PT) wins the bidding. 
	}
	
A \emph{repeated bidding game} is a recurrent game with three vertices: a root $v_0$ with two children $v_1$ and $v_2$.
\stam{ When \PLi wins the bidding at $v_0$, they proceed to $v_i$, for $i \in \set{1,2}$.}
	
	We derive an explicit expression for $\xb$ in such games. 
 First observe that the root $v_0$ is the only possible vertex with $\St(v_0)>0$. Thus, $\frac{St_i(v_0)}{\St_{max}}=1$. Therefore,  
 $\budget_i$ bids $\alpha_i x$ at $v_0$ (where $x$ is the budget).
    then from initial budget $x$, \PO wins the bidding if and only if $\alpha_1 x\ge \alpha_2 (1-x)$, equivalently $x\ge \frac{\alpha_2}{\alpha_1+\alpha_2}$. Then, the next budget if $\PO$ wins is $x-\alpha_1 x=(1-\alpha_1)x$, and otherwise $x+\alpha_2(1-x)=(1-\alpha_2)x + \alpha_2$. 
	We thus have the following piecewise-linear form.
	$
	\xb(x) = \begin{cases}
		(1 - \alpha_2) x + \alpha_2 & \text{ if } x < \frac{\alpha_2}{\alpha_1 + \alpha_2} \\
		(1 - \alpha_1) x & \text{ if } x \ge \frac{\alpha_2}{\alpha_1 + \alpha_2} \\
	\end{cases}
	$
\end{example}
Example~\ref{xmp:bowtie piecewise linear} can be extended by induction to general recurrent games 
to obtain the following. % (see the proof in App.~\ref{app:budget update is piecewise linear}).
%short , as follows. 
\begin{lemma}
\label{lem:budget update is piecewise linear}
    For every recurrent game $\G$ the budget-update function $\xb$ is piecewise-linear.
\end{lemma}
\begin{proof}
    Consider the tree representation of $\G$.
    We construct $\xb$ in a bottom-up fashion, starting from the leaves. Let $\St_{\max}$ be the maximal strength of a vertex in the game. 
    For each leaf $u$, we set the function $\xb_u(x)=x$. Indeed, no bidding takes place in the leaves, and the budget when returning to $v_0$ does not change.

    Next, assume $\xb_{u_1}$ and $\xb_{u_2}$ are piecewise linear functions already defined for subtrees of a subtree rooted at $u$, we define $\xb_u$ similarly to~\cref{xmp:bowtie piecewise linear}, as follows. Assume $u_1$ (resp. $u_2$) is the maximal potential child of $u$ for \PO (resp. \PT) (if both players share the potential maximizer, then no bidding take place and one of the subtrees can be removed).
    Then, according to $\budget_i$, \PLi bids $\alpha_i\frac{\St_{i}(u)}{\St_{\max}}x$ given budget $x$. Denote $\xi_i=\alpha_i\frac{\St_{i}(u)}{\St_{\max}}$
    Thus, \PO wins the bidding if and only if 
    $\xi_1 x\ge \xi_2(1-x)$, equivalently if 
    $x\ge \frac{\xi_2}{\xi_1+\xi_2}$. The budget transferred to $u_1$ if \PO wins is then $(1-\xi_1)x$, and otherwise $(1-\xi_2)x+\xi_2$ is transferred to $u_2$.

    It follows that the next budget at $v_0$ (starting from $u$ with budget $x$) is
    \[
    \xb_u(x)=\begin{cases}
		\xb_{u_2}((1 - \xi_2) x + \xi_2) & \text{ if } x < \frac{\xi_2}{\xi_1 + \xi_2} \\
		\xb_{u_1}((1 - \xi_1) x) & \text{ if } x \ge \frac{\xi_2}{\xi_1 + \xi_2} \\
	\end{cases}
    \]    
    $\xb_u$ is then piecewise linear as a composition (by branches) of piecewise-linear functions.

    We can then conclude the claim with $\xb\equiv \xb_{v_0}$.
\end{proof}

Intuitively, each linear ``branch'' of $\xb$ corresponds to an interval of budgets for which $\budget_1,\budget_2$ induce the same vine to the next milestone. By extension, notice that $\xb^i$ is also piecewise-linear for all $i$ (as a composition of piecewise-linear functions), and each branch of $\xb^i$ corresponds to a specific sequence of $i$ vines to be seen from an interval of budgets. 
We proceed to make this intuition formal. Define the set $\disc{}\subseteq [0,1]$ of \emph{bifurcation points} of $\xb$ to be the set of points where $\xb$ changes its behavior from one linear branch to another (e.g., in~\cref{xmp:bowtie piecewise linear} we have $\disc{}=\{\frac{\alpha_2}{\alpha_1+\alpha_2}\}$). We then extend this to $\xb^i$ for all $i\in \bbN$ by defining 
\[\disc{1}=\disc\quad\text{ and }\quad\disc{i}=\disc{i-1}\cup \{x\mid \xb(x)\in \disc{i-1}\}\]
By definition we have $\disc{i}\subseteq \disc{i+1}$. Thus, we can define the limit of this sequence as $\disc{\infty}=\bigcup_{i\in \bbN}\disc{i}$. Observe that $\disc{i}$ indeed captures the set of points where $\xb^i$ may change its linear branch. 
%The proof of the following lemma can be found in App.~\ref{app:vine fixed by linear branch}.

\begin{lemma}
\label{lem:vine fixed by linear branch}
Consider an interval $I\subseteq [0,1]$ such that $I\cap \disc{}=\emptyset$, then there is a vine $\theta$ such that for every $x\in I$, the vine taken first in $\path(x,\budget_1,\budget_2)$ is $\theta$.
\end{lemma}
% \stam{%short
 \begin{proof}
 	Intuitively, notice that the set of bifurcations in $\xb$ corresponds to the budget thresholds at every node in the tree of $\G$. It follows that all $x\in I$ induce the same vine in the tree. We proceed with the details, relying on the notation $\xb_u$ from the proof of~\cref{lem:budget update is piecewise linear}.

     Consider $x,y\in I$. Starting from $v_0$, since $I\cap \disc{}=\emptyset$, it follows that both $x$ and $y$ take the same bifurcation of $\xb_{v_0}$. This leads the play to either $v_1$ or $v_2$ (the children of $v_0$). Assume the game proceeds to $v_1$ (the case of $v_2$ is analogous). Then the budget with which $v_1$ is reached is $(1-\xi_1)x$ and $(1-\xi_1)y$ respectively. 
     Recall that in this branch we have $\xb_{v_0}(x)=\xb_{v_1}(1-\xi_1)x$ (and similarly for $y$). It follows that the bifurcations of $\xb_{v_0}$ arising in this branch are at all $b$ such that $(1-\xi_1)b$ is a bifurcation of $\xb_{v_1}$. 
     This means that $(1-\xi_1)x$ and $(1-\xi_1)y$ are again in a single branch of $\xb_{v_1}$, so we can continue the same reasoning by induction down to the leaves.
 \end{proof}
% }%of short

We now show that if $\disc{\infty}$ is finite, then $\path(c,\budget_1,\budget_2)$ is ultimately periodic. Intuitively, this is because if an interval $I$ with $I\cap \disc{\infty}=\empty$ is visited twice, we can show that due to the piecewise linearity, it is visited periodically. Then, Lem.~\ref{lem:vine fixed by linear branch} gives us the result. 
%See App.~\ref{app:budget periodic} for the proof and a depiction.
\begin{theorem}
\label{thm:disc finite then lasso}
Let $\path(c_0,\budget_1,\budget_2)=v_0,v_1,\ldots$. If $\disc{\infty}$ is finite, then there exists $N_0,k\in \bbN$ such that for every $n>N_0$ it holds that $v_n=v_{n+k}$.
\end{theorem}
\begin{proof}
	Since $\disc{\infty}$ is finite, let $N_1\in \bbN$ such that for all $n\ge N_1$ we have $\disc{n}=\disc{n+1}=\disc{\infty}$. 
	We then have $\disc{n}=\disc{n+1}=\disc{n}\cup \{x\mid \xb(x)\in \disc{n}\}$, and in particular $\{x\mid \xb(x)\in \disc{n}\}\subseteq \disc{n}$. In the contrapositive, we get that if $x\notin \disc{n}$, then $\xb(x)\notin \disc{n}$. 
	Therefore, for every interval $I\subseteq [0,1]$ such that $I\cap \disc{\infty}=\emptyset$ we have that $\xb(I)\cap \disc{\infty}=\emptyset$. 
	
	For every $i\in \bbN$, let $c_{i}=\xb^{i}(c_0)$ be the budget after $i$ milestones. For every $n\ge N_1$, let $I_n\subseteq [0,1]$ be the maximal interval (with respect to containment) such that $c_n\in I_n$ and $I_n\cap \disc{\infty}=\emptyset$ (i.e., $I_n$ is the ``entire branch'' of $\xb^n$ where $c_n$ lies). 
	Since $\disc{\infty}$ is finite, there exist $N_0\ge N_1$ and $k>0$ such that $I_{N_0}=I_{N_0+k}$. 
	Then, we claim that for every $n\ge N_0$ it holds that $I_{n}=I_{n+k}$, i.e., that the intervals visited after $N_0$ are visited periodically. 
	Indeed, since $c_{n}\in I_{n}$, then $c_{n+1}=\xb(c_{n})\in \xb(I_{n})$, but $\xb$ is linear when restricted to $I_{n}$, and therefore $\xb(I_{n})$ is an interval, and moreover $\xb(I_{n})\cap \disc{\infty}=\emptyset$ (since $n\ge N_1$ and therefore no additional bifurcations are introduced). 
 It follows that $\xb(I_{n})\subseteq I_{n+1}$. In particular, $\xb(I_{N_0+k})\subseteq I_{N_0+k+1}$ but also $\xb(I_{N_0+k})=\xb(I_{N_0})\subseteq I_{N_0+1}$. We depict this idea in~\cref{fig:intervals mapping}.
 Since the intervals under consideration are maximal, it follows that $I_{N_0+1}=I_{N_0+k+1}$, and we can proceed by induction to conclude the claim.

	Finally, by Lem.~\ref{lem:vine fixed by linear branch} we have that the ultimately periodic sequence of intervals induces an ultimately periodic sequence of vines visited from $N_0$, which concludes the proof.
\end{proof}
\begin{figure}[h]
    \begin{center}
        \begin{tikzpicture}[yscale=0.55]
            % Draw axes
            \draw[->] (0,0) -- (6,0) node[right] {$x$};
            \draw[->] (0,0) -- (0,5) node[left] {$\xb^{N_0}(x)$};
            
            % Draw the piecewise constant function
            \draw[thick, orange] (0,4 * 3/10) -- (1,4 * 3/10);
            \draw[thick, orange] (1,4 * 7/10) -- (2,4 * 7/10);
            \draw[thick, orange] (2,4 * 5/10) -- (3,4 * 5/10);
            \draw[thick, orange] (3,4 * 8/10) -- (4,4 * 8/10);
            \draw[thick, orange] (4,4 * 2/10) -- (5,4 * 2/10);
        
            % Vertical lines at steps
            \foreach \x in {1,2,3,4,5} {
                \draw[dotted] (\x,0) -- (\x,5);
            }
            
            % Labels for x-axis
            \foreach \x in {0, 0.2, 0.4, 0.6, 0.8, 1} {
                \draw (\x * 5, 0.2) -- (\x * 5, -0.2) node[below] {\x};
            }
            
            \draw[->, blue] 
            (2.5, 4 * 5/10) .. controls (3.5, 4 * 5/10 + 0.5) and (4,  4 * 5/10 + 0.5) .. (4.5, 4 * 2/10)
            node[pos=0.5, yshift=1ex, font=\scriptsize] {$\xb$};
        
            \draw[->, blue] 
            (4.55, 4 * 2/10) .. controls (4.2, 4 * 8/10 + 0.5) and (3.8,  4 * 8/10 + 0.5) .. (3.5, 4 * 8/10)
            node[pos=0.5, yshift=2ex, font=\scriptsize] {$\xb$};
        
            \draw[->, blue] 
            (3.45, 4 * 8/10) .. controls (3, 4 * 8/10 + 0.5) and (2,  4 * 8/10 + 0.5) .. (1.5, 4 * 7/10)
            node[pos=0.5, yshift=1ex, font=\scriptsize] {$\xb$};
        
            % return to the original point (2.5, 4 * 5/10)
            \draw[->, blue]
            (1.55, 4 * 7/10) .. controls (1.8, 4 * 7/10 + 0.3) and (2.3,  4 * 7/10 + 0.3) .. (2.45, 4 * 5/10)
            node[pos=0.5, yshift=1ex, font=\scriptsize] {$\xb$};   
        \end{tikzpicture}
    \end{center}
    \caption{ $\xb^{N_0}$ maps ``whole intervals'' to ``whole intervals'', and thus jumps between the intervals defined by $\disc{\infty}$ until returning to some interval, which is then mapped again into the cycle regardless of the specific budget with which the interval is reached.}
    \label{fig:intervals mapping}
\end{figure}
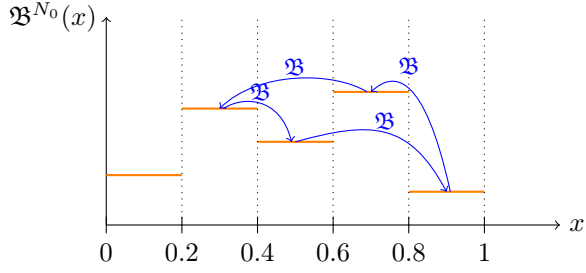

The proof of
%short %~\cref{thm:disc finite then lasso} 
Thm.~\ref{thm:disc finite then lasso} 
suggests a semi-algorithm for computing the mean payoff of the players in the play: given an initial budget $B_1$ keep track of $\disc{i}$ until it converges. If it does, keep track of the intervals until a period is reached. From there, evaluating the mean payoff amounts to evaluating the weights in the period of vines. 

\subsection{The generated path is ultimately periodic in repeated bidding games}
Analyzing the dynamics governed by piecewise-linear functions is notoriously difficult~\cite{pettit1997analysis}, and is typically handled only for simple cases, e.g.,~\cite{freire1998bifurcation,LN12}).
In this section, we show that $\disc{\infty}$ is finite for repeated bidding games (c.f.,~\cref{xmp:bowtie piecewise linear}), thus showing that the mean-payoff in this case  is computable.
\stam{
Analyzing the dynamics governed by piecewise-linear functions is notoriously difficult~\cite{pettit1997analysis}, and is typically handled only for simple cases, e.g.,~\cite{freire1998bifurcation,gaivao2024rotation}).
However, we are able to show that $\disc{\infty}$ is finite for repeated bidding games (c.f.,~\cref{xmp:bowtie piecewise linear}), thus showing that the mean-payoff in this case  is computable. We prove this in the remainder of this section. 
}

Recall from~\cref{xmp:bowtie piecewise linear} the form of $\xb(x)$ in a repeated bidding game. Specifically, we have that $\disc{1}=\{\frac{\alpha_2}{\alpha_1+\alpha_2}\}$. 
For every $i\in \bbN$ denote $A_i=\{x\mid \xb(x)\in \disc{i-1}\}$ (i.e., $A_i=\xb^{-1}(\disc{i-1})$), and note that $\disc{i}=\disc{i-1}\cup A_i$.
In order to prove that $\disc{\infty}$ is finite, it suffices to prove that $A_i=\emptyset$ for all $i\ge N_0$ for some $N_0$, or equivalently that $A_{N_0}=\emptyset$ for some $N_0$ (%after which 
so the claim follows by induction).

Consider the inverse relation $\xb^{-1}$ and some $x\in [0,1]$. In order to compute $\xb^{-1}(x)$ we use the explicit formulation of $\xb$: we see that $y\in \xb^{-1}(x)$ if either (1) $0\le y<\frac{\alpha_2}{\alpha_1+\alpha_2}$ and $x=(1-\alpha_2)y +\alpha_2$ or (2) $1\ge y\ge\frac{\alpha_2}{\alpha_1+\alpha_2}$ and $x=(1-\alpha_1)y$. 
%Rearranging the conditions, this is equivalent to (1) $y=\frac{x-\alpha_2}{1-\alpha_2}$ and $\alpha_2\le x<\frac{\alpha_2}{\alpha_1+\alpha_2}(1-\alpha_2)+\alpha_2$, or (2) $y=\frac{x}{1-\alpha_1}$ and $\frac{\alpha_2}{\alpha_1+\alpha_2}(1-\alpha_1)\le x\le 1-\alpha_1$.
Rearranging the above, we define the ``domains'' of the conditions as 
$I_1 = [ \alpha_2, \frac{\alpha_2}{\alpha_1 + \alpha_2} + \frac{\alpha_1 \alpha_2}{\alpha_1 + \alpha_2} )$ and $
        I_2 = [\frac{\alpha_2}{\alpha_1 + \alpha_2} - \frac{\alpha_1 \alpha_2}{\alpha_1 + \alpha_2}, 1 - \alpha_1]$

% \[
%         I_1 = \Bigg[ \alpha_2, \frac{\alpha_2}{\alpha_1 + \alpha_2} + \frac{\alpha_1 \alpha_2}{\alpha_1 + \alpha_2} \Bigg)\text{ and }
%         I_2 = \Bigg[\frac{\alpha_2}{\alpha_1 + \alpha_2} - \frac{\alpha_1 \alpha_2}{\alpha_1 + \alpha_2}, 1 - \alpha_1\Bigg]
%     \]
    We can then capture $\xb^{-1}$ using two functions:
\[R_1(x)=\frac{x-\alpha_2}{1-\alpha_2} \text{ if } x\in I_1 \text{ and } R_2(x)=\frac{x}{1-\alpha_1} \text{ if } x\in I_2\]
Indeed, we now have that $\xb^{-1}(x)=\{R_1(x), R_2(x)\}$ (by ignoring undefined values), and so $A_i=R_1(\disc{i-1})\cup R_2(\disc{i-1})$.

We now present a simple technical lemma from the theory of attractors in linear dynamical systems, % (see the proof in App.~\ref{app:distance from fixed point}), 
which we use to analyze the behavior of $R_1$ and $R_2$.

\begin{lemma}
    \label{lem:distance from fixed point}
    Let $y(x) = ax + b$ be a linear function and let $x_0\in \R$ be the (unique) fixed point of $y$, i.e. $y(x_0) = x_0$. Then, $\forall x\in \R: y(x) - x_0 = a(x - x_0)$. In particular, if $a > 1$ and $x\neq x_0$, then $\lim_{n\rightarrow \infty} |y^n(x) - x_0| = \infty$.
\end{lemma}
\begin{proof}
    First, we have that 
    \[y(x)-x_0=y(x)-y(x_0)=ax+b-ax_0-b=a(x-x_0).\]
    Next, observe that $x_0$ is also a fixed point of $y^n$, where $y^n(x)=a^nx+b'$ where $b'$ is some constant depending on $a,b,n$. Thus, if $a>1$, we have $y^n(x)-x_0=a^n(x-x_0)$, so 
    \[\lim_{n\to \infty} |y^n(x) - x_0|=\lim_{n\to \infty}a^n|x-x_0| = \infty.\]
\end{proof}

We are now ready to prove that $\disc{\infty}$ is finite.
\begin{theorem}
\label{thm:bowtie finite bif}
    In the notations above, for a repeated bidding game, $\disc{\infty}$ is finite. Moreover, 
    there exists $N_0=O\bigl(\max{\frac{\log{\frac{\alpha_i}{\alpha_i + \alpha_j}}}{\log{1 - \alpha_j}}}\bigr)$ for $i\neq j \in \{1,2\}$ such that $\disc{N_0}=\disc{\infty}$.
\end{theorem}
\begin{proof}
Notice that $R_1$ and $R_2$, thought of as $\bbR\to \bbR$ functions, are linear. Moreover, their unique fixed points are $x_0^1=1$ and $x_0^2=0$, respectively (indeed, $R_1(1)=1$ and $R_2(0)=0$). Consider some $x\in (0,1)$. By applying Lem~\ref{lem:distance from fixed point} to $R_1$ we get that 
$R_1(x)-x^1_0=\frac{1}{1-\alpha_2}(x-x^1_0)<x-x^1_0$ (since $\alpha_2\in (0,1)$ and $x-x^1_0<0$), and therefore $R_1(x)<x$, i.e., $R_1$ is ``left shifting''. Moreover, by Lem.~\ref{lem:distance from fixed point} we have that $\lim_{n\to\infty} R_1^n(x)=-\infty$. A similar analysis shows that $R_2$ is ``right shifting'' and $\lim_{n\to\infty} R_2^n(x)= \infty$.

For brevity, denote $\tau = \frac{\alpha_2}{\alpha_1 + \alpha_2}$ and $\chi = \frac{\alpha_1 \alpha_2}{\alpha_1 + \alpha_2}$. We show that there exists $N_0\in \bbN$ such that $A_{N_0}=\emptyset$. 
Intuitively, the idea is as follows: $R_1$ is left shifting, and its domain is below $\tau+\chi$. Therefore, any points in $A_i$ above $\tau+\chi$ can no longer decrease, and eventually escape beyond $1$, and therefore do not generate further bifurcation points. Similarly, $R_2$ is right shifting, and its domain is above $\tau-\chi$, so any points below it escape to below $0$. 

It is thus left to show 
%(see App.~\ref{app:bowtie finite bif})
 that after a finite number of iterations, points indeed go outside these intervals (See Fig.~\ref{fig:discontinuity dynamics}). 
We first show that $A_2\subseteq \{R_1(\tau),R_2(\tau)\}$ contains at most two points, one to the left of $I_2$ and one to the right of $I_1$. Since $R_2$ is right shifting and $R_1$ is left shifting, it follows that these two points eventually (after some $N_0$ iterations) escape the interval $(0,1)$, at which stage we have $A_{N_0}=\emptyset$.

%short Finally, the bound on $N_0$ follows by bounding the number of applications of $R_1$ (resp. $R_2$) to $\tau$ before it escapes the boundaries of $(0,1)$. Specifically, let $n$ be the number of iterations after which $R_2^n(\tau) \leq 1$, then by~\cref{lem:distance from fixed point} it holds that $\bigl(\frac{1}{1-\alpha_1}\bigr)^n \tau > 1$. Rearranging we get $n \leq \frac{\tau}{\log{(1 - \alpha_1)}}$, and by similar analysis for $R_1$ we get the upper bound. 

Recall that $\disc{1}=\{\tau\}$. We first show that if $\tau\in I_1$ then $R_1(\tau)$ falls to the left of $I_2$. Indeed:
\begin{align*}
        &\frac{\tau-\alpha_2}{1-\alpha_2}<\tau-\chi &\iff \\
        &\frac{\frac{\alpha_2}{\alpha_1 + \alpha_2} - \alpha_2}{1 - \alpha_2} < \frac{\alpha_2}{\alpha_1 + \alpha_2} - \frac{\alpha_1 \alpha_2}{\alpha_1 + \alpha_2}  & \iff \\
        &\frac{\alpha_2}{\alpha_1 + \alpha_2} - \alpha_2 < (1 - \alpha_2)(1 - \alpha_1)\bigl(\frac{\alpha_2}{\alpha_1 + \alpha_2}\bigr) & \iff \\
        & -1 < (\alpha_1 \alpha_2 - (\alpha_1 + \alpha_2)) \bigl(\frac{1}{\alpha_1 + \alpha_2}\bigr) & \iff \\
        &0 < \alpha_1 \alpha_2 & \text{which always holds.}
    \end{align*}
Similarly, if $\tau\in I_2$ then $R_2(\tau)$ falls to the right of $I_1$. Indeed:
 \begin{align*}
        &\frac{\tau}{1-\alpha_1}>\tau+\chi &\iff \\
        &\frac{\frac{\alpha_2}{\alpha_1 + \alpha_2}}{1 - \alpha_1} > \frac{\alpha_2}{\alpha_1 + \alpha_2} + \frac{\alpha_1 \alpha_2}{\alpha_1 + \alpha_2}  & \iff \\
        &\frac{\alpha_2}{\alpha_1 + \alpha_2} > (1 + \alpha_1)(1 - \alpha_1)\bigl(\frac{\alpha_2}{\alpha_1 + \alpha_2}\bigr) & \iff \\
        & 1 > 1 - \alpha_1^2  \iff \alpha_1^2>0 & \text{which always holds.}
    \end{align*}

    Finally, the bound on $N_0$ follows by bounding the number of applications of $R_1$ (resp. $R_2$) to $\tau$ before it escapes $(0,1)$. Let $n$ be the number of iterations after which $R_2^n(\tau) \leq 1$, then by Lem.~\ref{lem:distance from fixed point} it holds that $\bigl(\frac{1}{1-\alpha_1}\bigr)^n \tau > 1$. 
    Rearranging we get $n \leq \frac{\tau}{\log{(1 - \alpha_1)}}$, and by similar analysis for $R_1$ we get the upper bound.

\end{proof}
\begin{figure}[h]
\caption{The domains ${\color{red} I_1},{\color{blue} I_2}$ and the dynamics governed by ${\color{red} R_1},{\color{blue} R_2}$ starting from $\tau$. After a single iteration, ${\color{red} R_1}(\tau)$ is to the left of ${\color{blue} I_2}$, and ${\color{blue} R_2}(\tau)$ is to the right of ${\color{red} I_1}$. After several iterations, the points escape $(0,1)$.}
    \label{fig:discontinuity dynamics}
    \centering
    \begin{tikzpicture}
    \def\yt{0}
    \draw[ultra thick, red] (1,\yt + 0.05) -- (5,\yt + 0.05);
    \draw[ultra thick, blue] (3,\yt-0.05) -- (7,\yt-0.05);
    \draw[->, thick] (0,\yt) -- (8,\yt);
    \node[below] at (2,-0.3) {{\color{red} $I_1$}};
    \node[below] at (6,-0.3) {{\color{blue} $I_2$}};

    \node[below] at (4,\yt - 0.1) {$\tau$};
    \node[below] at (1,\yt - 0.1) {$\alpha_2$};
    \node[below] at (5.2,\yt - 0.1) {$\tau+\chi$};
    \node[below] at (2.8,\yt - 0.1) {$\tau-\chi$};
    \node[below] at (7,\yt - 0.1) {$1 - \alpha_1$};

    \node at (1,\yt) {\large(};
    \node at (5,\yt) {\large)};
    \node at (3,\yt) {\large(};
    \node at (7,\yt) {\large)};

    \draw (4,\yt + 0.1) -- (4,\yt - 0.1);

    \draw[->, blue] 
    (4,\yt) .. controls (4.5,\yt + 1) and (5,\yt + 1) .. (5.5,\yt+0.05)
    node[pos=0.5, yshift=1ex, font=\scriptsize] {$R_2$};

    \draw[->, blue] 
        (5.5,\yt) .. controls (5.5 + 1/3, \yt + 1) and (5.5 + 2/3,\yt + 1) .. (6.5, \yt+0.05)
        node[pos=0.5, yshift=1ex, font=\scriptsize] {$R_2$};
    \draw[->, blue] 
    (6.5,\yt) .. controls (6.5 + 1/3, \yt + 1) and (6.5 + 2/3,\yt + 1) .. (7.5, \yt+0.05)
    node[pos=0.5, yshift=1ex, font=\scriptsize] {$R_2$};

    \draw[->, red] 
        (4,\yt) .. controls (3.5,\yt + 1) and (3,\yt + 1) .. (2.5,\yt + 0.05)
        node[pos=0.5, yshift=1ex, font=\scriptsize] {$R_1$};
    
    \draw[->, red] 
        (2.5,\yt) .. controls (2.5 - 1/3, \yt + 1) and (2.5 - 2/3,\yt + 1) .. (1.5, \yt+0.05)
        node[pos=0.5, yshift=1ex, font=\scriptsize] {$R_1$};

    \draw[->, red]
        (1.5,\yt) .. controls (1.5 - 1/3, \yt + 1) and (1.5 - 2/3,\yt + 1) .. (0.5, \yt+0.05)
        node[pos=0.5, yshift=1ex, font=\scriptsize] {$R_1$};
    
    \end{tikzpicture}
    \end{figure}
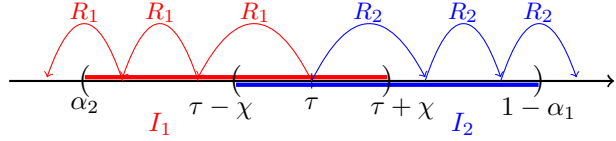

    By Thms.~\ref{thm:disc finite then lasso} and \ref{thm:bowtie finite bif} we have that for repeated bidding games,
    %~\cref{alg:budget-MP} 
    the algorithm always halts, and we have a bound on $N_0$. This enables us to conclude the following.
    \begin{corollary}
    \label{cor:bowtie ptime}
    Given a repeated bidding game, 
    %the value 
    $\MP_i(\pi)$ is computable in time polynomial in $O(\max{\frac{\log{\frac{\alpha_i}{\alpha_i + \alpha_j}}}{\log{1 - \alpha_j}}})$ for $i\neq j \in \{1,2\}$. 
    \end{corollary}

%%%%%%%%%%%%%%%%%%%%%%%%%%%%%%%%    
\section{Discussion}
We analyze the path generated by the two types of known explicit optimal strategies in mean-payoff bidding games, show that it is ultimately periodic in certain graphs, and develop algorithms to reason about its performance. Interestingly, the discrete nature of the block strategy is an advantage in our study whereas in the study of {\em all-pay} bidding~\cite{AJZ21}, for which the budget strategies were first constructed, the continuous nature of the latter was essential. 
We list several directions for future work. 
First, we leave open questions; e.g., analyzing 
%showing that the path is ultimately periodic in 
general strongly-connected graphs and tightening the computational complexity.
% of reasoning about the performance of the generate play. 
Second, analyzing the generated play is a general phenomenon and can be studied in any game. The analysis is technically interesting in mean-payoff bidding games since optimal strategies are succinctly represented, and it is interesting to analyze games with other bidding mechanisms for which explicit constructions are known~\cite{AHI18,AJZ21}. 
%Particularly since there, the initial budget allocation is crucial. 
Finally, the strategies that we consider are designed offline to optimize against an adversary. An interesting direction for future research is to develop strategies that adapt online to the competitor they are paired with. We expect that the novel analysis techniques that we develop here will be useful to develop such strategies.

\subsection*{Acknowledgments}
This work was supported in part by the Israel Science Foundation, grant numbers 989/22 and 1679/21.


% Bibliography is supplied in observed.bbl for arXiv.

% Appendix
\appendix

\section{Vertex importance; strengths and potentials}
\label{app:strengths}
The \emph{strength} of a vertex, denoted $\St: V \rightarrow \Q$, intuitively measures its importance; if $\St(v) < \St(u)$, then winning a bidding in $u$ is more important than in $v$. 

\begin{definition}{\bf (Potentials and strengths).}
    The \emph{potential} is a function $\Pot:V \rightarrow \Q$ that for every vertex $v\in V$, satisfies 
    %$\Pot(v) = \frac{1}{2}(\Pot(v^+)+\Pot(v^-)) + w(v) - \MP(\RT(\G))$
     \[
         \Pot(v) = \frac{\Pot(v^+)+\Pot(v^-)}{2} + w(v) - \MP(\RT(\G))
         \]
where $v^+ = \arg\max_{u\in \neig(v)} \Pot(u)$ and $v^- = \arg\min_{u\in \neig(v)} \Pot(u)$. For $v \in V$, the \emph{strength} is a non-negative quantity
based on the 
% defined as the difference between the maximal and minimal 
potential in $\neig(v)$. Specifically:
    %    $\St(v) = \frac{1}{2}(\Pot(v^+)-\Pot(v^-))$.
         \[
         \St(v) = \frac{\Pot(v^+)-\Pot(v^-)}{2}
     \]
\end{definition}
The existence of a potential function is guaranteed following results on stochastic games~\cite{Put05}. 
\begin{example}
\label{ex:strengths}
We illustrate the idea behind the definition of strengths. The details, which are orthogonal to this paper, can be found in \cite{AHC19}. 
Consider the game $\G$ that is depicted in Fig~\ref{fig:block}. 
We have $\MP\big(\RT(\G)\big) = 0$. 
It is not hard to verify that the potentials and strengths satisfy the requirements. For example, 
$\Pot(v_0) = \frac12 \cdot (\Pot(v_1) + \Pot(\ell_3)) = \frac{1}{2} \cdot (-2 + 2) = 0$ and $\St(v_0) = \frac12 \cdot (\Pot(\ell_3) - \Pot(v_1)) = \frac{1}{2} \cdot (2- (-2)) = 2$. 
\Max's strategy ties between changes in energy and changes in budget. 
To illustrate, assume that there are no budget requirements, \Max bids $\St(v)$ at $v \in V$, and follows the outgoing red edge upon winning the bidding. We describe two \Min responses. First, \Min wins the bidding at $v_0$, proceeds to $v_1$, and pays \Max at least $2$ units of budget, then \Max wins the bidding at $v_1$, pays \Min $3$ units of budget, and proceeds to $\ell_1$, where the energy increases by $1$. All in all, \Max ``invests'' $1$ unit of budget in order to increase the energy by $1$. Second, \Min wins both biddings, i.e., she pays \Max at least $2$ at $v_0$ and $3$ at $v_1$ for a total of $5$. The game reaches $\ell_2$ where the energy decreases by $5$. Thus, \Max ``gains'' $5$ units of budget when the energy decreases by $5$. One can verify that the last vine satisfies this property.
\end{example}

\end{document}